\documentclass[11pt]{llncs}

\usepackage{latexsym}

\usepackage{fullpage}
\usepackage[natbib=true, sortcites=true, style=alphabetic]{biblatex}
\bibliography{delay}
\usepackage{amssymb,amsmath,amsfonts,bm}
\usepackage{subfigure} 
\usepackage{epic} 
\usepackage{eepic} 
\usepackage{enumerate}
\usepackage{vinayak}
\usepackage{graphicx}
\usepackage{verbatim}
\usepackage[mathscr]{eucal}
\usepackage{xspace}
\usepackage[boxed]{algorithm2e}
 \SetAlgoInsideSkip{smallskip}
\setalcapskip{9pt}
 \SetAlFnt{\normalsize\rm}
\usepackage{enumerate}
\usepackage{paralist}

\title{Finite Automata with Time-Delay Blocks
\thanks{%
This work has been financially supported in part by the
European Commission FP7-ICT Cognitive Systems, Interaction, and Robotics under the contract  \# 270180 (NOPTILUS);  
by Funda\c c\~ ao para Ci\^ encia e Tecnologia under project 
PTDC/EEA-CRO/104901/2008 (Modeling and control of Networked vehicle systems in persistent autonomous
operations);
by Austrian Science Fund (FWF) Grant No P
23499-N23 on Modern Graph Algorithmic Techniques in Formal Verification;
FWF NFN Grant No S11407-N23 (RiSE); ERC Start grant (279307: Graph Games);
Microsoft faculty fellows award;  ERC Advanced grant QUAREM; and FWF Grant
No  S11403-N23 (RiSE).}
}
\author{Krishnendu Chatterjee$^1$ \and Thomas A. Henzinger$^{1}$   
\and Vinayak S. Prabhu$^2$ }
\institute{
$^1$ IST Austria (Institute of Science and Technology Austria); $\qquad$
$^2$FEUP, University of Porto; \\
{\tt \{krish.chat,tah\}@ist.ac.at, vinayak@fe.up.pt}
}
\date{}

\begin{document}
\maketitle

\thispagestyle{empty}

\begin{abstract}
The notion of delays arises naturally in many computational models, such as, in the 
design of circuits, control systems, and dataflow languages.
In this work, we introduce \emph{automata with delay blocks} (ADBs), extending 
finite state automata with variable time delay blocks, 
for deferring individual transition output symbols, in a discrete-time setting.
We show that the ADB languages strictly subsume the regular languages, and 
are incomparable in expressive power to the context-free languages.
We show that ADBs are closed under union, concatenation and Kleene star, 
 and under intersection with regular 
languages, but not closed under complementation and intersection with other
ADB languages.
We show that the emptiness and the membership problems are decidable in 
polynomial time for ADBs, whereas the universality problem is undecidable.
Finally we consider the linear-time model checking problem, i.e., whether the 
language of an ADB is contained in a regular language, and show that the model 
checking problem is PSPACE-complete.
\end{abstract}

%\vspace*{-15mm}
\section{Introduction}
The class of dynamical systems (or processes) with delays occur frequently in 
control systems where delays arise due to physical constraints
(see \emph{e.g.}~\cite{zhong2006robust, chiasson2007applications, 
sipahi2012time, Gu_Kharitonov_Chen_2003}). 
The notion of delays is also common in systems where transmission of 
information is involved. 
Delay blocks have been used for modeling such time delays in engineering 
systems,
for example, the unit delay block in Simulink~\cite{SimulinkDelay} delays the 
input signal by one 
sample period, corresponding to the $z^{-1}$ discrete time $Z$-transform 
operator.
The memory block in Simulink, meant for continuous time signals, delays the 
input by one integration time step. 
Mathworks' Control Systems Toolbox~\cite{ControlToolbox} can be used for modeling delays in control 
systems using the $ e^{-\Delta s}$
Laplace transform operator (in the transfer functions)
for modeling a delay of  $\Delta$ time units; 
the coupling between the delay and the
system dynamics is tracked in the internal state space model.
The notion of delays arises naturally in other computational models, 
e.g., time delays are used in the design and analysis of 
circuits (timing analysis and analysis of circuits with latches),
and delays are a key component in dataflow languages
(\emph{e.g.} in the Ptolemy II framework~\cite{EkerJanneckLeeLiuLiuLudvigSachsXiong03_TamingHeterogeneityPtolemyApproach, PtolemyOnline}).

Although delay constructs have been widely used in control systems, 
design of circuits, and dataflow languages, they have not been considered in 
the classical automata theoretic settings in computer science.
One approach to model delays in the automata theoretic setting has been by the
introduction of an automaton model for an intermediate buffer for
explicitly modeling the state of the buffer. 
This approach suffers from three crucial drawbacks: 
(1)~the buffer length has to be fixed in any given model, 
(2)~the buffer contents have to be explicitly modeled leading to
unnecessary model complexity, and 
(3)~the state space of the system blows up with increasing buffer size, due
to state space modeling of the buffer contents.

In this work, we introduce an extension of the standard finite state automata 
model by enriching automata with variable discrete-time delay blocks for
deferring individual output symbols. 
We call the resultant structures \emph{automata with delay blocks} (ADBs). 
Viewing the automata as generators of strings, 
%in the standard finite state 
%automata model 
the string generated by an accepting run of a standard finite state automaton
is the 
same as the sequence of symbols observed as the output of the run. 
In automata with delay blocks, the output symbols are 
\emph{generated} by a regular 
automaton structure, but the \emph{output sequence} of the symbols  
differs from 
the symbol generation sequence due to the delay blocks involved. 
In an ADB, there is an associated discrete-time delay $\Delta$ with each 
transition $e$ labelled by an output symbol; in the output the symbol
labeling the edge $e$ appears after a delay of $\Delta$  time units. 
Time passes in the model in discrete time steps, either via an explicit 
$\tick$ transition in the ADB, or when the automaton run ends in an accepting state. 
We present a couple of examples to illustrate the model. 
Given an ADB $\A$, let $\lan(\A)$ denote the (discrete-time) output language 
of the automaton, and let $\ulan(\A)$ denote the untimed output language.

\begin{example}
Consider a shipwreck scenario  where hazardous material containers from a wrecked ship 
 are floating in the ocean, and are being dispersed by ocean currents.
A team of autonomous underwater vehicles (AUVs) is monitoring the situation, their
goal being to (1)~detect the possible locations of the drums using sonar data, and
(2)~monitor affects on underwater marine life due to leaking materials from the
containers. 
For illustrative purposes, consider a team of two vehicles named AUV-1 and AUV-2.
AUV-1 is operating at a depth of 10 meters, and is taking sonar imaging data above it 
and processing it to detect the floating drum locations.
 AUV-2 is operating at a depth of 150 meters and monitoring the
underwater marine life situation.
The search pattern of AUV-2 depends on the possible sightings of containers given by
AUV-1 which are conveyed   through acoustic communication.
AUV-1 periodically, surfaces as it is close to the surface, sends its full detailed imaging 
data to the base station through GSM communication (high datarate and only works 
above water, underwater acoustic communication is extremely 
low datarate and has limited range) 
 where human operators study data and update the earlier
sighting inferences of AUV-1, and send the updates back to AUV-1, which must
then convey the updates back to AUV-2 through underwater acoustic communication.
The human  operators may also change the resurfacing frequency of AUV-1 depending
on the data received.

We are interested in describing the pattern of messages received by AUV-2.
We define one discrete time unit to be the time in between two AUV-1 resurfacings (note
that this corresponds to variable physical times, and a variable number of point 
monitorings).
In one such time unit, AUV-1 sends $k$ point locations to AUV-2, each annotated with 
$\yes$ and 
$\may$ (for possible
container sightings, $\may$ denotes ``maybe'').
The updates from the base station  are conveyed to AUV-2 in the next time slot 
from AUV-1 as simply  a $k$-bit sequence
corresponding to the same $k$ locations as in the previous time slot (the locations are
not sent again to AUV-2 as underwater communication is extremely expensive).
Let us denote the sending of the point coordinates as the event $\point$.
Then the (untimed) language describing the pattern of messages from AUV-1 to AUV-2 is
\[
\set{ (w\# w')  \mid w \in \set{\point  \yes, \point   \may}^* \text{ and }
w' \in \set{\yes, \no}^* \text{ and } 2|w'| = |w|}^*
\]
 where
$\#$ denotes the demarcation between two adjacent time slots.
This language can be described in a natural and intuitive fashion by the ADB in 
Figure~\ref{figure:example-zero}.
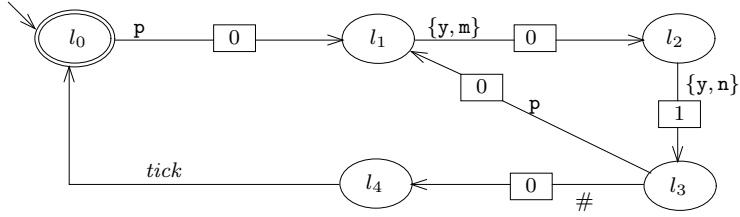
\begin{figure}[t]
\hspace*{-4mm}
  \strut\centerline{\setlength{\unitlength}{0.00043745in}
\begingroup\makeatletter\ifx\SetFigFont\undefined%
\gdef\SetFigFont#1#2#3#4#5{%
  \reset@font\fontsize{#1}{#2pt}%
  \fontfamily{#3}\fontseries{#4}\fontshape{#5}%
  \selectfont}%
\fi\endgroup%
{\renewcommand{\dashlinestretch}{30}
\begin{picture}(8396,2530)(0,-10)
\put(799,2076){\ellipse{854}{584}}
\put(799,2076){\ellipse{944}{674}}
\put(7948,2074){\ellipse{854}{584}}
\put(4389,2070){\ellipse{854}{584}}
\put(4365,357){\ellipse{854}{584}}
\put(7961,332){\ellipse{854}{584}}
\path(2442,2233)(2937,2233)(2937,1918)
	(2442,1918)(2442,2233)
\path(1272,2053)(2442,2053)
\path(2937,2053)(3972,2053)
\path(3792.000,1993.000)(3972.000,2053.000)(3792.000,2113.000)
\path(4827,2053)(5997,2053)
\path(6492,2053)(7527,2053)
\path(7347.000,1993.000)(7527.000,2053.000)(7347.000,2113.000)
\path(12,2503)(327,2188)
\path(157.294,2272.853)(327.000,2188.000)(242.147,2357.706)
\path(7707,1333)(8202,1333)(8202,1018)
	(7707,1018)(7707,1333)
\path(7932,1783)(7932,1333)
\path(7932,1018)(7932,613)
\path(7872.000,793.000)(7932.000,613.000)(7992.000,793.000)
\path(3927,343)(732,343)(732,1738)
\path(792.000,1558.000)(732.000,1738.000)(672.000,1558.000)
\path(5952,478)(6447,478)(6447,163)
	(5952,163)(5952,478)
\path(7527,343)(6447,343)
\path(5952,343)(4782,343)
\path(4962.000,403.000)(4782.000,343.000)(4962.000,283.000)
\path(5367,1648)(5862,1648)(5862,1333)
	(5367,1333)(5367,1648)
\path(5862,1333)(7617,478)
\path(5367,1648)(4782,1918)
\path(4970.576,1897.047)(4782.000,1918.000)(4920.289,1788.092)
\path(5997,2233)(6492,2233)(6492,1918)
	(5997,1918)(5997,2233)
\put(732,2008){\makebox(0,0)[lb]{\smash{{\SetFigFont{8}{9.6}{\familydefault}{\mddefault}{\updefault}$l_0$}}}}
\put(2622,2008){\makebox(0,0)[lb]{\smash{{\SetFigFont{8}{9.6}{\familydefault}{\mddefault}{\updefault}0}}}}
\put(6132,2008){\makebox(0,0)[lb]{\smash{{\SetFigFont{8}{9.6}{\familydefault}{\mddefault}{\updefault}0}}}}
\put(1497,2143){\makebox(0,0)[lb]{\smash{{\SetFigFont{8}{9.6}{\familydefault}{\mddefault}{\updefault}$\point$}}}}
\put(8022,1468){\makebox(0,0)[lb]{\smash{{\SetFigFont{8}{9.6}{\familydefault}{\mddefault}{\updefault}$\set{\yes,\no}$}}}}
\put(7887,1108){\makebox(0,0)[lb]{\smash{{\SetFigFont{8}{9.6}{\familydefault}{\mddefault}{\updefault}1}}}}
\put(1632,433){\makebox(0,0)[lb]{\smash{{\SetFigFont{8}{9.6}{\familydefault}{\mddefault}{\updefault}$\tick$}}}}
\put(4287,2008){\makebox(0,0)[lb]{\smash{{\SetFigFont{8}{9.6}{\familydefault}{\mddefault}{\updefault}$l_1$}}}}
\put(7797,2008){\makebox(0,0)[lb]{\smash{{\SetFigFont{8}{9.6}{\familydefault}{\mddefault}{\updefault}$l_2$}}}}
\put(6132,253){\makebox(0,0)[lb]{\smash{{\SetFigFont{8}{9.6}{\familydefault}{\mddefault}{\updefault}0}}}}
\put(4242,298){\makebox(0,0)[lb]{\smash{{\SetFigFont{8}{9.6}{\familydefault}{\mddefault}{\updefault}$l_4$}}}}
\put(7842,253){\makebox(0,0)[lb]{\smash{{\SetFigFont{8}{9.6}{\familydefault}{\mddefault}{\updefault}$l_3$}}}}
\put(5547,1423){\makebox(0,0)[lb]{\smash{{\SetFigFont{8}{9.6}{\familydefault}{\mddefault}{\updefault}0}}}}
\put(6177,1243){\makebox(0,0)[lb]{\smash{{\SetFigFont{8}{9.6}{\familydefault}{\mddefault}{\updefault}$\point$}}}}
\put(6717,73){\makebox(0,0)[lb]{\smash{{\SetFigFont{8}{9.6}{\familydefault}{\mddefault}{\updefault}$\#$}}}}
\put(4962,2143){\makebox(0,0)[lb]{\smash{{\SetFigFont{8}{9.6}{\familydefault}{\mddefault}{\updefault}$\set{\yes,\may}$}}}}
\end{picture}
}}
  \caption{Automaton $\A_0$ with delay blocks.}
  \label{figure:example-zero}
\end{figure}
The automaton also makes it clear that the $i$-th  $\yes,\no$ that appears
after the $\#$ corresponds to the $i$-th $\point \yes, \point\may$ in the previous
time slot; this relationship may be useful for further processing of  the
point coordinates values.  
We explain the workings of the automaton $\A_0$ in detail below.

The initial location  is $l_0$ which is also the only accepting location.
Each edge has a delay block, with the number inside the block denoting the time 
delay associated with the block.
Consider accepting runs of the automaton. 
The output symbols are \emph{generated} in accepting runs according to the regular 
expression  sequence $(\point \set{\yes, \may} \set{\yes, \no} )^*$ (the
transition labeled $\tick$ denotes time passing by one time unit and
$\tick$  is not an
output symbol).
However, because of the associated delays with the transitions, the output 
symbols
(namely $\point, \yes, \may, \no$) appear in a different sequence.
Consider a particular run sequence $r = \point \yes \yes \# \tick  \point \yes \no
\point \may\no\# \tick$.
Recall that time advances in ADBs either via the explicit
$\tick$ transition, or when the run ends in an accepting  state. 
Thus, in the run $r$, the first four symbols (\emph{i.e.} $\point \yes \yes \#$) are
generated at time $0$.
The second $\yes$ symbol has an associated delay of $1$, the 
rest have an associated delay of $0$.
The $0$-delay symbols appear immediately in the output (at time $0$).
Then, we have the first $\tick$ transition, which results in time 
advancing to $1$.
At time $1$, first the $1$-delay symbol, $\yes$ (generated previously) appears at the output.
Then, the sequence  $\point \yes \no \point \may\no\#$ is generated, with the
first and the second $\no$ symbols having a delay of $1$.
Except for these two delayed $\no$ symbols, the rest appear immediately
at time $1$.
Then comes the second  $\tick$ transition 
 which results in time advancing to $2$, and at time $2$, the 
two delayed $\no$ symbols appear.
Thus, the time stamped output sequence corresponding to the run $r$ after time $2$ 
is
$\tuple{\point, 0} \tuple{\yes, 0} \tuple{\#,0} \tuple{\yes, 1}
\tuple{\point, 1} \tuple{\yes, 1} \tuple{\point, 1} \tuple{\may,1}
\tuple{\#, 1} \tuple{\no, 2} \tuple{\no,2}$
(the second element in the tuples denotes the timestamp when the first element of
the tuple appears in the output).
\qed
\end{example}

\begin{example}
 Consider the ADB $\A_1$ in Figure~\ref{figure:example-one}.
\begin{figure}[t]
  \strut\centerline{\setlength{\unitlength}{0.00043745in}
\begingroup\makeatletter\ifx\SetFigFontNFSS\undefined%
\gdef\SetFigFontNFSS#1#2#3#4#5{%
  \reset@font\fontsize{#1}{#2pt}%
  \fontfamily{#3}\fontseries{#4}\fontshape{#5}%
  \selectfont}%
\fi\endgroup%
{\renewcommand{\dashlinestretch}{30}
\begin{picture}(4882,2691)(0,-10)
\put(844,1697){\ellipse{854}{584}}
\put(2617,299){\ellipse{854}{584}}
\put(4447,1688){\ellipse{854}{584}}
\put(844,1697){\ellipse{944}{674}}
\path(1137,999)(1632,999)(1632,684)
	(1137,684)(1137,999)
\path(2307,2664)(2802,2664)(2802,2349)
	(2307,2349)(2307,2664)
\path(3612,999)(4107,999)(4107,684)
	(3612,684)(3612,999)
\path(3612,684)(2982,369)
\path(3116.164,503.164)(2982.000,369.000)(3169.830,395.833)
\path(1137,954)(912,1449)
\path(1041.107,1309.962)(912.000,1449.000)(931.863,1260.306)
\path(2802,2484)(2805,2484)(2810,2483)
	(2820,2481)(2836,2478)(2858,2475)
	(2886,2470)(2919,2464)(2957,2457)
	(2999,2450)(3044,2442)(3090,2434)
	(3137,2425)(3183,2416)(3227,2407)
	(3270,2399)(3311,2390)(3350,2382)
	(3386,2373)(3420,2365)(3452,2357)
	(3482,2349)(3511,2341)(3538,2332)
	(3563,2324)(3588,2315)(3611,2306)
	(3635,2296)(3659,2286)(3683,2275)
	(3706,2263)(3730,2251)(3754,2237)
	(3778,2223)(3802,2207)(3828,2191)
	(3854,2172)(3882,2153)(3910,2132)
	(3939,2109)(3970,2086)(4000,2062)
	(4031,2037)(4061,2013)(4089,1990)
	(4115,1968)(4138,1949)(4157,1933)
	(4173,1920)(4197,1899)
\path(4022.026,1972.376)(4197.000,1899.000)(4101.046,2062.685)
\path(4287,1404)(3927,999)
\path(822,1989)(824,1991)(828,1995)
	(834,2003)(845,2015)(859,2030)
	(876,2049)(897,2070)(919,2093)
	(942,2116)(967,2139)(991,2162)
	(1016,2183)(1040,2202)(1065,2220)
	(1089,2237)(1113,2252)(1138,2265)
	(1164,2278)(1191,2290)(1219,2301)
	(1250,2311)(1271,2318)(1294,2325)
	(1318,2332)(1343,2338)(1369,2345)
	(1398,2351)(1428,2358)(1460,2364)
	(1495,2370)(1532,2377)(1571,2384)
	(1613,2390)(1657,2397)(1704,2404)
	(1754,2411)(1805,2419)(1857,2426)
	(1910,2433)(1964,2440)(2016,2447)
	(2066,2454)(2114,2460)(2157,2465)
	(2195,2470)(2228,2474)(2255,2478)
	(2275,2480)(2290,2482)(2300,2483)
	(2305,2484)(2307,2484)
\path(1317,684)(1319,682)(1325,678)
	(1334,670)(1348,659)(1366,644)
	(1387,628)(1411,609)(1436,589)
	(1461,570)(1485,552)(1509,534)
	(1532,518)(1553,503)(1573,490)
	(1593,478)(1612,466)(1631,456)
	(1650,446)(1670,436)(1688,428)
	(1706,420)(1726,412)(1746,404)
	(1769,396)(1793,388)(1818,379)
	(1846,370)(1877,361)(1909,352)
	(1943,342)(1978,332)(2014,322)
	(2049,312)(2081,303)(2109,296)
	(2133,289)(2150,285)(2162,282)
	(2169,280)(2172,279)
\path(12,1674)(372,1674)
\path(192.000,1614.000)(372.000,1674.000)(192.000,1734.000)
\put(777,1629){\makebox(0,0)[lb]{\smash{{\SetFigFontNFSS{8}{9.6}{\familydefault}{\mddefault}{\updefault}$l_0$}}}}
\put(4287,1629){\makebox(0,0)[lb]{\smash{{\SetFigFontNFSS{8}{9.6}{\familydefault}{\mddefault}{\updefault}$l_1$}}}}
\put(2487,234){\makebox(0,0)[lb]{\smash{{\SetFigFontNFSS{8}{9.6}{\familydefault}{\mddefault}{\updefault}$l_2$}}}}
\put(3837,729){\makebox(0,0)[lb]{\smash{{\SetFigFontNFSS{8}{9.6}{\familydefault}{\mddefault}{\updefault}1}}}}
\put(1362,729){\makebox(0,0)[lb]{\smash{{\SetFigFontNFSS{8}{9.6}{\familydefault}{\mddefault}{\updefault}2}}}}
\put(2487,2439){\makebox(0,0)[lb]{\smash{{\SetFigFontNFSS{8}{9.6}{\familydefault}{\mddefault}{\updefault}0}}}}
\put(1182,2394){\makebox(0,0)[lb]{\smash{{\SetFigFontNFSS{8}{9.6}{\familydefault}{\mddefault}{\updefault}$a$}}}}
\put(4242,1044){\makebox(0,0)[lb]{\smash{{\SetFigFontNFSS{8}{9.6}{\familydefault}{\mddefault}{\updefault}$b$}}}}
\put(1407,279){\makebox(0,0)[lb]{\smash{{\SetFigFontNFSS{8}{9.6}{\familydefault}{\mddefault}{\updefault}$c$}}}}
\end{picture}
}}
  \caption{Automaton $\A_1$ with delay blocks.}
  \label{figure:example-one}
\end{figure}
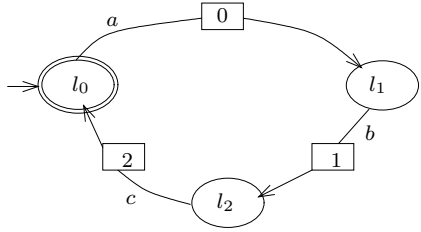
The initial state is $l_0$, which is also the only accepting state. 
Consider an accepting run of the automaton. 
The output symbols are generated in accepting runs according to the regular 
expression  sequence $(abc)^*$. 
However, the output delay associated with the transition for $a$ is $0$, 
for $b$  is $1$, and for $c$ the delay is $2$. 
As there are no explicit time advancing $\tick$ transitions, time advances 
only when the run ends in the accepting state, and then the symbols with 
delay $0$ are observed (according to their generation sequence), then the
symbols at time $1$, and so on. It can be seen that the output symbol
sequence for the ADB $\A_1$ is  $a^n b^n c^n$.
Thus, the untimed language $\ulan (\A_1 )$  is 
$\set{a^n b^n c^n \mid  n \geq 0}$. 
Including the output time stamps in the words, we
get the timed language $\lan(\A_1 )$ as  
$\set{\tuple{a,0}^n \tuple{b,1}^n \tuple{c,2}^n \mid  n \geq 0}$
(no output symbols appear after time $2$).
\qed
\end{example}

\begin{example}
Consider the ADB $\A_2$ in Figure~\ref{figure:example-two}.
\begin{figure}[t]
  \strut\centerline{\setlength{\unitlength}{0.00043745in}
\begingroup\makeatletter\ifx\SetFigFontNFSS\undefined%
\gdef\SetFigFontNFSS#1#2#3#4#5{%
  \reset@font\fontsize{#1}{#2pt}%
  \fontfamily{#3}\fontseries{#4}\fontshape{#5}%
  \selectfont}%
\fi\endgroup%
{\renewcommand{\dashlinestretch}{30}
\begin{picture}(6669,2691)(0,-10)
\put(2631,1697){\ellipse{854}{584}}
\put(6234,1688){\ellipse{854}{584}}
\put(2631,1697){\ellipse{944}{674}}
\put(435,1696){\ellipse{854}{584}}
\put(4404,299){\ellipse{854}{584}}
\path(2924,999)(3419,999)(3419,684)
	(2924,684)(2924,999)
\path(4094,2664)(4589,2664)(4589,2349)
	(4094,2349)(4094,2664)
\path(5399,999)(5894,999)(5894,684)
	(5399,684)(5399,999)
\path(2249,2619)(2429,1989)
\path(2321.859,2145.591)(2429.000,1989.000)(2437.242,2178.558)
\path(5399,684)(4769,369)
\path(4903.164,503.164)(4769.000,369.000)(4956.830,395.833)
\path(2924,954)(2699,1449)
\path(2828.107,1309.962)(2699.000,1449.000)(2718.863,1260.306)
\path(4589,2484)(4592,2484)(4597,2483)
	(4607,2481)(4623,2478)(4645,2475)
	(4673,2470)(4706,2464)(4744,2457)
	(4786,2450)(4831,2442)(4877,2434)
	(4924,2425)(4970,2416)(5014,2407)
	(5057,2399)(5098,2390)(5137,2382)
	(5173,2373)(5207,2365)(5239,2357)
	(5269,2349)(5298,2341)(5325,2332)
	(5350,2324)(5375,2315)(5398,2306)
	(5422,2296)(5446,2286)(5470,2275)
	(5493,2263)(5517,2251)(5541,2237)
	(5565,2223)(5589,2207)(5615,2191)
	(5641,2172)(5669,2153)(5697,2132)
	(5726,2109)(5757,2086)(5787,2062)
	(5818,2037)(5848,2013)(5876,1990)
	(5902,1968)(5925,1949)(5944,1933)
	(5960,1920)(5984,1899)
\path(5809.026,1972.376)(5984.000,1899.000)(5888.046,2062.685)
\path(6074,1404)(5714,999)
\path(2609,1989)(2611,1991)(2615,1995)
	(2621,2003)(2632,2015)(2646,2030)
	(2663,2049)(2684,2070)(2706,2093)
	(2729,2116)(2754,2139)(2778,2162)
	(2803,2183)(2827,2202)(2852,2220)
	(2876,2237)(2900,2252)(2925,2265)
	(2951,2278)(2978,2290)(3006,2301)
	(3037,2311)(3058,2318)(3081,2325)
	(3105,2332)(3130,2338)(3156,2345)
	(3185,2351)(3215,2358)(3247,2364)
	(3282,2370)(3319,2377)(3358,2384)
	(3400,2390)(3444,2397)(3491,2404)
	(3541,2411)(3592,2419)(3644,2426)
	(3697,2433)(3751,2440)(3803,2447)
	(3853,2454)(3901,2460)(3944,2465)
	(3982,2470)(4015,2474)(4042,2478)
	(4062,2480)(4077,2482)(4087,2483)
	(4092,2484)(4094,2484)
\path(3104,684)(3106,682)(3112,678)
	(3121,670)(3135,659)(3153,644)
	(3174,628)(3198,609)(3223,589)
	(3248,570)(3272,552)(3296,534)
	(3319,518)(3340,503)(3360,490)
	(3380,478)(3399,466)(3418,456)
	(3437,446)(3457,436)(3475,428)
	(3493,420)(3513,412)(3533,404)
	(3556,396)(3580,388)(3605,379)
	(3633,370)(3664,361)(3696,352)
	(3730,342)(3765,332)(3801,322)
	(3836,312)(3868,303)(3896,296)
	(3920,289)(3937,285)(3949,282)
	(3956,280)(3959,279)
\path(2159,1629)(2157,1627)(2151,1623)
	(2142,1617)(2127,1606)(2107,1592)
	(2083,1575)(2055,1555)(2023,1533)
	(1990,1510)(1956,1487)(1922,1464)
	(1889,1442)(1857,1421)(1827,1402)
	(1799,1384)(1772,1368)(1746,1354)
	(1722,1341)(1699,1329)(1676,1318)
	(1655,1308)(1633,1300)(1612,1291)
	(1588,1283)(1564,1276)(1540,1269)
	(1515,1263)(1490,1258)(1465,1253)
	(1438,1249)(1412,1245)(1385,1242)
	(1358,1240)(1332,1239)(1305,1238)
	(1279,1238)(1254,1238)(1229,1239)
	(1205,1241)(1182,1243)(1159,1246)
	(1138,1249)(1118,1253)(1098,1257)
	(1079,1261)(1059,1267)(1039,1273)
	(1019,1280)(999,1287)(979,1296)
	(958,1306)(936,1317)(913,1329)
	(889,1343)(864,1357)(839,1373)
	(813,1388)(790,1403)(768,1417)
	(750,1429)(719,1449)
\path(902.781,1401.835)(719.000,1449.000)(837.726,1300.999)
\path(719,1944)(722,1946)(728,1949)
	(739,1955)(754,1964)(775,1976)
	(799,1989)(826,2004)(855,2019)
	(884,2034)(912,2049)(939,2062)
	(965,2074)(990,2084)(1013,2093)
	(1036,2101)(1058,2108)(1080,2114)
	(1102,2120)(1124,2124)(1145,2128)
	(1166,2131)(1189,2133)(1212,2135)
	(1236,2137)(1261,2138)(1287,2138)
	(1314,2138)(1341,2137)(1369,2136)
	(1397,2134)(1425,2132)(1454,2128)
	(1482,2125)(1509,2121)(1537,2116)
	(1563,2111)(1589,2105)(1614,2099)
	(1639,2093)(1663,2086)(1687,2079)
	(1708,2072)(1729,2064)(1751,2056)
	(1772,2047)(1795,2037)(1818,2027)
	(1842,2015)(1868,2002)(1895,1988)
	(1924,1972)(1953,1956)(1984,1939)
	(2016,1920)(2048,1902)(2079,1884)
	(2108,1867)(2134,1851)(2156,1838)
	(2175,1827)(2204,1809)
\path(2019.423,1852.947)(2204.000,1809.000)(2082.707,1954.904)
\put(2564,1629){\makebox(0,0)[lb]{\smash{{\SetFigFontNFSS{8}{9.6}{\familydefault}{\mddefault}{\updefault}$l_0$}}}}
\put(5624,729){\makebox(0,0)[lb]{\smash{{\SetFigFontNFSS{8}{9.6}{\familydefault}{\mddefault}{\updefault}1}}}}
\put(3149,729){\makebox(0,0)[lb]{\smash{{\SetFigFontNFSS{8}{9.6}{\familydefault}{\mddefault}{\updefault}2}}}}
\put(4274,2439){\makebox(0,0)[lb]{\smash{{\SetFigFontNFSS{8}{9.6}{\familydefault}{\mddefault}{\updefault}0}}}}
\put(2969,2394){\makebox(0,0)[lb]{\smash{{\SetFigFontNFSS{8}{9.6}{\familydefault}{\mddefault}{\updefault}$a$}}}}
\put(6029,1044){\makebox(0,0)[lb]{\smash{{\SetFigFontNFSS{8}{9.6}{\familydefault}{\mddefault}{\updefault}$b$}}}}
\put(3194,279){\makebox(0,0)[lb]{\smash{{\SetFigFontNFSS{8}{9.6}{\familydefault}{\mddefault}{\updefault}$c$}}}}
\put(1169,2169){\makebox(0,0)[lb]{\smash{{\SetFigFontNFSS{8}{9.6}{\familydefault}{\mddefault}{\updefault}$\tick$}}}}
\put(1034,954){\makebox(0,0)[lb]{\smash{{\SetFigFontNFSS{8}{9.6}{\familydefault}{\mddefault}{\updefault}$\tick$}}}}
\put(359,1629){\makebox(0,0)[lb]{\smash{{\SetFigFontNFSS{8}{9.6}{\familydefault}{\mddefault}{\updefault}$l_3$}}}}
\put(6164,1629){\makebox(0,0)[lb]{\smash{{\SetFigFontNFSS{8}{9.6}{\familydefault}{\mddefault}{\updefault}$l_1$}}}}
\put(4319,234){\makebox(0,0)[lb]{\smash{{\SetFigFontNFSS{8}{9.6}{\familydefault}{\mddefault}{\updefault}$l_2$}}}}
\end{picture}
}}
  \caption{Automaton $\A_2$ with delay blocks.}
  \label{figure:example-two}
\end{figure}
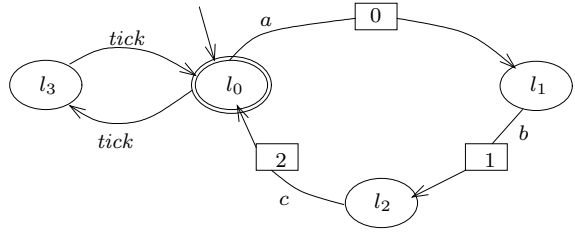
The initial state is $l_0$, which is also the only accepting state. 
The                                                                                          
accepting runs of the automaton correspond to the regular expression sequence 
$\left( abc \left(\tick \tick\right)^*\right) ^*$. 
Consider a particular
run  sequence $r  = abc\, abc\, \tick  \tick abc \tick \tick \tick \tick abc$. 
Recall that time advances in ADBs either via the explicit
$\tick$ transition, or when the run ends in an accepting  state. 
In the run $r$, the first two $a$ occurrences are generated at time
0 as no $\tick$ transitions have been encountered until then;
these two $a$ occurrences appear immediately in the
output at time 0 (the associated delay is $0$ for the delay block). 
The first two $b$ occurrences are also generated at time $0$, but
appear in the output at time $1$, when the first $\tick$ transition is taken. 
The first two $c$ occurrences are generated
at time $0$, and appear in the output at time $2$, when the second $\tick$ transition is 
taken. 
Thus, after the first two $\tick$ transitions, the time-stamped output string is 
$\tuple{a, 0}^2 \tuple{b, 1}^2 \tuple{c, 2}^2 $. 
The third $a$ occurrence in $r$ is generated
at time $2$ (after the first two $\tick$ transitions), and appears immediately at time $2$. 
The third $b$ occurrence in $r$ is generated at time $2$, and appears
after a delay of one time unit, when the third $\tick$ transition is taken. 
The third $c$ occurrence in $r$ is generated at
time $2$, and appears at time $4$, when the fourth $\tick$ transition is taken. 
Continuing in this fashion, we see that
the time-stamped output corresponding to the run $r$ is 
$\tuple{a, 0}^2 \tuple{b, 1}^2 \tuple{c, 2}^2\,  \tuple{a, 2} \tuple{b, 3}  \tuple{c, 4}
\tuple{a, 6} \tuple{b, 7}  \tuple{c, 8}$. 

    Letting , $\tuple{\sigma,i}^0$ denote the empty string, the timed language of the 
automaton $\A_2$  can be observed to be
\[
\lan(\A_2) = 
\left \{ %\left.
\begin{array}{c}
  \tuple{a, 0}^{n_0} \tuple{b, 1}^{n_0} \tuple{c, 2}^{n_0}\,
  \tuple{a, 2}^{n_2} \tuple{b, 3}^{n_2} \tuple{c, 4}^{n_2} \dots\\
   \tuple{a, 2k}^{n_{2k}} \tuple{b, 2k+1}^{n_{2k}} \tuple{c, 2k+2}^{n_{2k}}\\
%\ \right\vert \ 
   \text{such that}\\
              n_i\geq  0 \text{ for all } i, \text{  and } k  \geq 0
\end{array}
\right\}
\]
The untimed language of the automaton 
$\A_2$ can be observed to be
\[
\ulan(\A_2 ) = \left\{\left.\mspace{-10mu}
\begin{array}{l}
a^{n_0} b^{n_0} c^{n_0}  
  a^{n_1} b^{n_1} c^{n_1} \dots\\
 \qquad a^{n_k} b^{n_k} c^{n_k} 
\end{array}
\right\vert\ 
  n_i \geq 0 \text{ for all } i \text{ and  } k \geq 0
\right\}
\]
Equivalently, using the untimed language of the automaton $\A_1$ from the 
previous example,
$\ulan(\A_2 ) = \set{w_0 w_1 \dots w_n \mid w_i \in \ulan(\A_1) \text{ for } 0\leq i\leq n}$.
\qed
\end{example}

\noindent{\bf Our contributions.}
In this work, along with the introduction of ADBs, we study their expressive
power, closure properties, and the basic decision and model checking problems.
Our main results are as follows:
\begin{compactenum}[$\star$]
\item \emph{Expressive power:}
We show that the untimed languages of ADBs strictly subsume regular 
languages, and are incomparable in expressive power to context-free languages.
ADBs are able to express a simple class of languages not expressible by
context-free languages.
For example, the automata $\A_1$ of Figure~\ref{figure:example-one} has the
untimed language $\set{a^n b^n c^n \mid  n \geq 0}$ which is not context free.
\item \emph{Closure properties:} We show that untimed ADB languages
 are closed under union, concatenation,  Kleene star, and 
intersection with regular languages, but not under complementation and 
intersection with other untimed ADB languages.
\item \emph{Decision and model checking problems:}
We show that the emptiness and the membership problems are decidable for ADBs
in polynomial time, 
whereas the universality of untimed ADB languages is undecidable.
Finally, we consider the model checking problem, where an ADB is 
considered as the model generating words, and a regular language specifies
the desired set of words. 
The model checking problem is then the containment of the untimed ADB language in 
the regular language, and we show that the problem is 
PSPACE-complete.
\end{compactenum}
Thus, ADBs provide a natural and practical extension of finite state automata 
for modeling discrete time processes involving delays where the output generation 
is via a regular process.
ADBs though incomparable in expressiveness to context-free languages, enjoy 
several nice properties similar to that of context-free languages, 
for instance, ADBs admit decidable emptiness, membership and model checking 
algorithms.
We note that the delays used in ADBs are of most use in \emph{modelling}
and \emph{analysis} of
naturally occurring delays in physical systems, not in directly \emph{building}
engineering systems.
Thus, non-closure under intersection of ADBs is not a deal-breaker --- systems
are built compositionally as regular automata; delays are only used in the analysis of the composed system.

For our technical contribution we present illustrative ideas behind two of 
the key results.
(1)~We show that the balanced parenthesis language is not expressible as an
untimed ADB language.
This is a bit surprising because ADBs can express non context-free languages 
like $a^n b^n c^n$.
This inexpressibility (which establishes incomparability to context-free
languages) is a result of the fact that the maximum delay
present in an ADB  limits the ``depth''  of the nestings in the generated
word.
Consider a word $a^n \circ w \circ  b ^n$, where $\circ$ is the concatenation
operator, and $w$ is a subword.
To match the $a^n$ with the $b^n$, the ADB needs to use at least one
delay block, say of delay $k$.
Then, to express matchings in the word $w$, it can only use delay blocks
of delay \emph{strictly less than} $k$.
(2)~We can model check an untimed ADB language against  a regular
specification (\emph{i.e.} a finite state automaton).
To show this, we check for emptiness of an untimed
ADB language and a regular language complement of the
specification by constructing a
non-deterministic finite state  automaton which has an
accepting path iff the intersection of the languages is non-empty.
This automaton maintains a guess of the future executions of the
regular specification automaton for $M$ future timepoints, where
$M$ is the largest delay of the given ADB.
The guesses are verified whenever time advances.
Proofs omitted from the main paper can be found in the appendix.

%Thus we solve the model checking problem.
%If $M$ is a constant, and the specification is given as a deterministic
%finite state automaton, then model checking can be done
%in polynomial time.

\smallskip\noindent\textbf{Related Work.}
The model of timed automata~\cite{AlurD94}
is a widely studied formalism for timed systems.
Timed automata do not have any construct for delaying 
generated output symbols, and their untimed languages are regular, unlike
for ADBs.
In the task scheduling context, a model which is somewhat related has 
recently been introduced in~\cite{StiggeEGY11}, the digraph real-time task model
 (DRT).
In a DRT instance, jobs are released according to a specified directed weighted graph, where the weights on the edges denote the time that must elapse
in between the job releases.
The nodes, which correspond to jobs, are  annotated with
the worst case execution times and the deadlines for the jobs.
Thus, the deadline  sequence for when the jobs must finish differs from
the jobs release sequence 
due to the deadline and execution time ``delays''.
However, the edge weights in the DRT model are \emph{strictly} positive and
integer valued --- this implies that the ``queue'' of currently executing jobs
has length at most $N$ where $N$ can be computed from the DRT instance.
Thus, the deadline sequences form a regular set.
In ADBs, an \emph{unbounded} number of symbols can be generated, before
an output symbol is seen, thus the implicit queue is of unbounded length.
This additional power of ADBs
 can be used to model scheduling problems where a bound on 
the number  of job creations per unit time is not known a priori.
The work in~\cite{MalerP95} explores delays in circuits.
 Only delays signals which "hold" for a given time $d$ are of relevance, where 
$d$ is a given constant; signals which do not persist for at least $d$ time units
are not output. This gives regularity, allowing
the system to be modeled as a timed automaton.
We do not require a hold time, in our discrete time framework, an unlimited number  
of letters (actually all) in between two
time ticks are delayed if so specified.

The ADB model also has some similarity to computational models of
automata augmented with queues.
An ADB with $M$ delay blocks can be viewed as writing to $M$ unbounded queues at any 
given point in time, corresponding to the $M$ delays indexed by the delay 
blocks.
The work in~\cite{Ibarra00} presents decidability results for reader-writer 
systems augmented with one unbounded queue in between the reader and writer for 
communication, one pushdown stack for either the reader or writer, and 
finitely many reversal bounded counters for both.
It also shows undecidability for two finite state automata (reader and writer) 
with two unbounded communication queues in between. 
The work of~\cite{Brand1983} shows decidability results for two finite state 
automata augmented with an unbounded one way communication queue in between 
them, and mention undecidability if there are more than two communicating 
finite state automata augmented with just one unbounded queue in the system.
The work of~\cite{Bouajjani98symbolicreachability} presents symbolic 
semi-algorithms for analyzing communicating finite state automata with queue 
communication channels.
If the queue channels are \emph{lossy}, then decidability can be shown for a 
variety of problems~\cite{AbdullaJ96a}.
Model checking is usually done on systems with \emph{bounded} buffers (see 
\emph{e.g.}~\cite{FrehseM07}), and suffers from the state explosion problem with
increasing buffer size.
Our key result shows that the ADB model has the decidable model checking 
property in spite of containing any number of unbounded \emph{delay} buffers.
One key intuition behind the decidable result is the fact that messages 
corresponding to time $\Delta$ are invisible to an observer until all messages 
corresponding to the previous time-points have been output and consumed.

%%{\bf KRISH TO VP: Short description of two proofs: maybe incomparable with 
%%PDA, and model checking decidable.}

\begin{comment}
We note that many ``regular'' expression libraries (\emph{e.g.} the \emph{re} 
regular expressions module in the Python Standard Library, Perl, GNU egrep) 
support \emph{back-referencing} in their ``regular'' expressions --- 
back-referencing remembers a matched substring and allows its use later in 
the ``regular'' expression.
For example, backreferencing allows matches to non-regular patterns like 
$b^n a b^n a b^n$ (this pattern is expressed as 
\texttt{(b*)a\textbackslash1a\textbackslash1}).
Patterns like $b^n a c^n a d^n$  however  cannot be expressed.
The delay block construct in ADBs allows us to implement and extend the 
expressiveness of  back-referencing in these ``regular'' expression patterns.
\end{comment}

%%% Local Variables: 
%%% mode: latex
%%% TeX-master: t
%%% End: 

\section{Automata with Delay Blocks} 
\label{section:Model}
In this section we introduce our model of automata with delay blocks,
and illustrate with examples the timed and untimed languages generated
by these automata.

\smallskip\noindent{\bf Automata with delay blocks (ADB).}
A finite \emph{automata with delay blocks} (ADB) is a tuple 
$\A=(L,D,\Sigma, \delta, l_s, L_f)$ where
\begin{compactitem}
\item
  $L$ is a finite set of locations.
\item
  $l_s\in L$ is the starting location.
\item
  $L_f\subseteq L$ is the set of accepting locations.
\item
  $\Sigma$ is the set of output symbols.
  %We denote $\Sigma\cup \set{\epsilon}$ by $\Sigma^{\epsilon}$, with
  %$\epsilon$ denoting the empty string.
\item
  $D$ is a finite set of delay blocks.
  Each delay block $d\in D$ is indexed by a natural number $t\geq 0$ 
  to indicate the
  amount of delay for the outputs.
  We denote a delay block with delay $t$ by $\block{t}$.
\item
  $\delta$ is the transition relation,  
\[
  \delta: \, \Big( \left(L\times\Sigma\times D\right)\
    \cup\
    \left(L\times\set{\epsilon,\tick} \right)\Big)
    \mapsto 2^L\]
  where  $\epsilon$ denotes the empty string, and
  $\tick\notin \Sigma$ denotes a time passage of one time unit.
%  The different types of transitions are explained bl
  \begin{compactitem}
  \item
    A transition $\delta(l,\sigma, \block{t}) = L'$ denotes a location change
    from $l$ to a location in $L'$ non-deterministically, with $\sigma$ being output $t$ time units into the future.
  \item
     A transition $\delta(l,\epsilon) = L'$ denotes an epsilon transition
     from $l$ to a location in $L'$ non-deterministically, with no new output requirements.
   \item
      A transition $\delta(l,\tick) = L'$ denotes time advancing by one time
      unit, and the location changing from $l$ to a location in $L'$ non-deterministically.
  \end{compactitem}
\end{compactitem}
A finite string $w$ is a sequence of elements. Given a string $w$, we let $|w|$ 
denote the length of the string $w$, and let $w[i]$ denote the $i$-th element 
(starting from index 0) in the string $w$ if $|w| >  i$.
The empty string is denoted by $\epsilon$.
The concatenation of two words $w_1$  and $w_2$ is denoted $w_1 w_2$ and also
$w_1\circ w_2$.
We also use the standard regular expression constructs.
%; and the set 
%$\set{w_1,w_2}$ is denoted $w_1 + w_2$. {\bf KRISH: CHANGE THE NOTATION WITH $+$.}
For $i\geq 0$, we denote by $\rep_i(w)$ the string $w$ repeated $i$ times (letting
$\rep_0(w) = \epsilon$),
\emph{i.e.,} $\rep_i(w)$ is the string $\underset{i \text{ occurences}}{\underbrace{w w\dots w}}$.

\smallskip\noindent\textbf{Discrete Timed words.}
A (discrete) timed word $w$ is a finite string belonging to $(\Sigma\times\nat)^*$ where
$\nat$ denotes the set of natural numbers.
We refer to the first element of the tuple $w[i]$ as the \emph{output 
symbol}
and 
the second element of the tuple $w[i]$ as the \emph{timestamp}. 
The timestamps denote the discrete time at which the first element of the 
tuples appear in the word.
We require that for $i < j < |w|$, and for 
 $w[i]=\tuple{w_i^{\sigma}, w_i^{t}}$ and  
$w[j]=\tuple{w_j^{\sigma}, w_j^{t}}$, we have $ w_i^{t} \leq  w_j^{t}$ (i.e. the 
timestamps are non-decreasing).
Given a timed word $w\in (\Sigma\times\nat)^*$, let $\untime(w)\in \Sigma^*$
be the untimed word
denoting the projection of $w$ onto $\Sigma^*$, that is, if
$w=\tuple{\sigma_0,t_0} \dots \tuple{\sigma_m,t_m}$, then 
$\untime(w) = \sigma_0\dots\sigma_m$.
Given  a timed word $w=\tuple{\sigma_0, t_0}\dots \tuple{\sigma_n, t_n}$ and
a natural number $\Delta \geq 0$, we
let $w\oplus \Delta$ be the timed word
$\tuple{\sigma_0, t_0+\Delta}\dots \tuple{\sigma_n, t_n+\Delta}$ (the time stamps
are advanced by $\Delta$ for all $w[i]$).
  Given an untimed word $w = \sigma_0 \sigma_1 \dots \sigma_m$,
  let $\kappa_t(w)$ denote the timed word
  $\tuple{\sigma_0, t} \tuple{\sigma_1, t} \dots \tuple{\sigma_m, t}$,
  that is the timed word where each output symbol of $w$ occurs at time $t$.

\smallskip\noindent\textbf{Generation of discrete timed words by ADBs.}
A generating \emph{run} $r$ of the automaton $\A$ is a finite sequence 
$l_0 \stackrel{\alpha_0}{\longrightarrow} l_1  
\stackrel{\alpha_1}{\longrightarrow}\dots l_n$ for
$\alpha_i \in \set{\epsilon,\tick}\, \cup\, \left(\Sigma\times D\right)$, such that
$l_0$ is the starting location, $l_n$ is an accepting location
  and $l_{i+1} \in \delta(l_i,\alpha_i) $ for 
$0\leq i\leq n-1$.
The automaton $\A$ \emph{outputs} or \emph{generates} 
the timed word $w$ if there exists a 
generating run
$l_0 \stackrel{\alpha_0}{\longrightarrow} l_1  
\stackrel{\alpha_1}{\longrightarrow}\dots l_n$ 
%with $l_n\in L_f$ and
%$l_0 = l_s$ 
such that
$\oword(\alpha_0 \dots \alpha_n) = w$ where, 
informally,  the $\oword()$ function timestamps the output symbols according
to their generation  and  delay block times,
and arranges them in the proper timestamp order.
 A delay block $\block{j}$ 
delays the output symbol by $j$ time units.
At time $t\in \nat$ in a run, a  delay block $\block{j}$ can be considered to be feeding 
symbols to a 
queue $\q_{t+j}$ which will output the stored symbols  at time $t+j$ (there
is only one queue corresponding to an output time $t$).
A $\tick$ transition explicitly advances time by one time unit.
We also have that once the automaton stops at a final state, 
time automatically advances with symbols stored in the queues being output at the
appropriate times.
We note that time advances \emph{only} at $\tick$ transitions, or when the automaton
comes to rest at a final state.

Formally, $\oword(\alpha_0 \dots \alpha_n) $ is the unique timed word
$w$  belonging to $(\Sigma\times\nat)^*$ defined as follows.
For $\overline{\alpha} = \tuple{\alpha_0 \dots \alpha_n}$, and 
$\tuple{\alpha_i^{\sigma}, \block{{\alpha_i^t}} } \in \set{\alpha_0,\dots, \alpha_n}$,
let
$\wtime(\tuple{\alpha_i^{\sigma}, \block{\alpha_i^t}}, \overline{\alpha}) = 
\alpha_i^t + t_i$, where
$t_i$ denotes the number of occurrences of $\tick$ in
$\alpha_0,\dots, \alpha_{i-1}$.
Intuitively, the $\sigma$-element $\alpha_i^{\sigma}$  of each  
$\tuple{\alpha_i^{\sigma}, \block{\alpha_i^t}} \in \set{\alpha_0,\dots, \alpha_n}$,
appears exactly once in $w$,
with $\wtime(\tuple{\alpha_i^{\sigma}, \block{\alpha_i^t}}, \overline{\alpha}) $
denoting its timestamp.
Formally,  $\oword(\alpha_0 \dots  \alpha_n) $ is the unique timed word
$w$ such that
\begin{compactitem}
\item 
 $|w|$ is equal to the number of times symbols from $\Sigma\times \nat$ 
 appear in the string $\alpha_0 \alpha_1 \dots \alpha_n$.
\item 
For all $i < |w|$, we have
  $w[i]=
\tuple{\alpha_j^{\sigma}, \wtime(\tuple{\alpha_j^{\sigma}, \block{\alpha_j^t}}, \overline{\alpha}) }$
where $\tuple{\alpha_j^{\sigma}, \block{\alpha_j^t}} = \alpha[j]$ is such that
 for  all $k$ and $\alpha[k]  = \tuple{\alpha_k^{\sigma}, \block{\alpha_k^t}}$,
the following conditions hold.
\begin{compactenum}
\item 
If either 
\begin{compactitem}

\item 
$k < j$ and $\wtime(\alpha[k] , \overline{\alpha}) \leq  \wtime(\alpha[j] , \overline{\alpha}) $;
or
\item $k > j$ and  $\wtime(\alpha[k] , \overline{\alpha}) < \wtime(\alpha[j] , \overline{\alpha}) $,
\end{compactitem}
then for some $i' < i$, we have 
$w[i'] =  \tuple{\alpha_k^{\sigma}, \wtime(\alpha[k], \overline{\alpha}) }$.
\item
If either
\begin{compactitem}

\item 
$ k < j$ and $\wtime(\alpha[k] , \overline{\alpha}) > \wtime(\alpha[j] , \overline{\alpha}) $;
or 
\item $k > j$ and  $\wtime(\alpha[k] , \overline{\alpha}) \geq  \wtime(\alpha[j] , \overline{\alpha}) $,
\end{compactitem}
then for some $i' >  i$, we have 
$w[i'] =  \tuple{\alpha_k^{\sigma}, \wtime(\alpha[k], \overline{\alpha}) }$.
\end{compactenum}
Thus, the placement of the $\sigma$-element of each $\alpha[j]$ is in 
increasing order of the timestamps  $\wtime(\alpha[j] , \overline{\alpha})$,
and if  $\alpha[j]$ and $\alpha[k]$ result in the same timestamp,
then the relative ordering is dictated by the relative ordering
between $j$ and $k$.
\end{compactitem}

An equivalent alternative algorithmic definition of the function  $\oword()$ is
given  in Function~\ref{function:Oword} with $\stablesorttime$ being a stable
sorting function which sorts based on the second element of 
tuples \footnote{Stable sorting algorithms maintain the original 
relative ordering of elements
 with equal key values.}.

%
%
%
%\SetAlCapSkip{3mm}
\begin{function}
 %    \SetKw{Let}{let}
     \SetKwInOut{Input}{Input}
    \SetKwInOut{Output}{Output}
    \Input{A string $\alpha$ from 
      $\left(\left(\Sigma\times D\right)\, \cup\, \set{\epsilon,\tick}\right)^*$}
    \Output{A timed word $w$ in $(\Sigma\times \nat)^*$}
%    $\oword(\alpha_0,\dots, \alpha_n)$:\\
    $w=\epsilon$\;
    $\currtime=i=j=0$\;
    \While{$i < |\alpha|$}
    {
      \Switch{$\alpha[i]$}
      {
        \uCase{$\epsilon$}
        {
          $i:= i+1$\;
        }
        \uCase{$\tick$}
        {
          $i:= i+1$\;
          $\currtime := \currtime+1$\;
        }
        \uCase{$\tuple{\sigma,\block{m}}$}
        {
          $w[j]=\tuple{\sigma, \currtime+m}$\;
          $i:= i+1$\;
          $j:= j+1$\;
        }
      }
    }
    \Return{$\stablesorttime(w)$}\;
%\BlankLine
%\BlankLine
%\BlankLine
%\AlCapSkip{1}
%\setalcapskip{length}{2mm}
  \caption{
$\oword$($\alpha$)}
  \label{function:Oword}
%Note the format in the caption, function, arguments
\end{function}
%
%
%

% Intuitively, a delay block $\block{j}$ 
% delays the output symbol by $j$ time units.
% At time $t\in \nat$, a  delay block $\block{j}$ can be considered to be feeding 
% symbols to a 
% queue $\q_{t+j}$ which will output the stored symbols  at time $t+j$ (there
% is only one queue corresponding to an output time $t$).
% A $\tick$ transition explicitly advances time by one time unit.
% We also have that once the automaton stops at a final state, 
% time automatically advances with symbols stored in the queues being output at the
% appropriate times.
% We note that time advances \emph{only} at $\tick$ transitions, or when the automaton
% comes to rest at a final state.

\smallskip\noindent\textbf{Output languages of ADBs.}
The timed output \emph{language} of  $\A$ is denoted by
$\lan(\A)$ where
$\lan(\A) = \set{ w \mid w\text{ is a timed word generated by } \A}$.
For a  timed language $\lan$, we let $\untime(\lan) =
\set{ \untime(w) \mid w\in \lan}$.
We also let $\ulan(\A) $ denote the untimed language 
$\untime(\lan(\A))$.
We have already illustrated languages of ADBs with two 
examples in the introduction. 
Below we present another illustrating example.

\begin{comment}
\begin{example}

  \begin{figure}[t]
    \begin{minipage}[b]{0.5\linewidth}
      \strut\centerline{\input Figures/example1.eepic}
      \caption{Automaton $\A_1$ with delay blocks.}
      \label{figure:example-one}
    \end{minipage}
    \begin{minipage}[b]{0.5\linewidth}
      \strut\centerline{\input Figures/example2.eepic}
      \caption{Automaton  $\A_2$ with delay blocks.}
      \label{figure:example-two}
    \end{minipage}
    
    \end{figure}
  % 
  Consider the ADBs $\A_1$ and $A_2$ 
  in Figures~\ref{figure:example-one} and~\ref{figure:example-two} with 
  the initial states
  $l_0$ which are also the only final states.
  The timed language $\lan(\A_1)$ is 
  $\set{\tuple{a,0}^n \tuple{b,1}^n \tuple{c,2}^n \mid n\geq 0}$, and the 
  untimed
  language $\ulan(\A_1)$ is $\set{a^n b^n c^n \mid n\geq 0}$.
  The timed language $\lan(\A_2)$ is
  $$\lan(\A_2) = 
  \set{(w_0)(w_1 \oplus 2)(w_2 \oplus 4)\dots(w_n \oplus 2 n) \mid 
    n \geq 0 \text{ and } 
    w_i \in \lan(A_1) \text{ for } 0\leq i \leq n}$$
  The untimed  language of $\A_2$ is
  $$\ulan(\A_2) = \set{ w_0^u w_1^u \dots w_n^u \mid 
     n \geq 0 \text{ and } 
    w_i^u \in \ulan(A_1) \text{ for } 0\leq i \leq n}$$ 
{\bf KRISH: EXPLAIN CLEARLY FOR BOTH EXAMPLES HOW WE OBTAIN THE TIMED 
AND UNTIMED LANGUAGE.}
\qed
\end{example}
\end{comment}

%%length
\begin{example}
  \begin{figure}[t]
    \strut\centerline{\setlength{\unitlength}{0.00043745in}
\begingroup\makeatletter\ifx\SetFigFontNFSS\undefined%
\gdef\SetFigFontNFSS#1#2#3#4#5{%
  \reset@font\fontsize{#1}{#2pt}%
  \fontfamily{#3}\fontseries{#4}\fontshape{#5}%
  \selectfont}%
\fi\endgroup%
{\renewcommand{\dashlinestretch}{30}
\begin{picture}(6372,2149)(0,-10)
\put(3372,1041){\ellipse{944}{674}}
\put(3372,1041){\ellipse{854}{584}}
\put(435,1072){\ellipse{854}{584}}
\put(5937,996){\ellipse{854}{584}}
\path(1595,1873)(2090,1873)(2090,1558)
	(1595,1558)(1595,1873)
\path(4385,613)(4880,613)(4880,298)
	(4385,298)(4385,613)
\path(1595,613)(2090,613)(2090,298)
	(1595,298)(1595,613)
\path(4385,1873)(4880,1873)(4880,1558)
	(4385,1558)(4385,1873)
\path(2945,1153)(2943,1156)(2940,1163)
	(2934,1174)(2924,1191)(2911,1214)
	(2896,1242)(2878,1273)(2859,1306)
	(2839,1340)(2818,1373)(2798,1405)
	(2779,1435)(2760,1462)(2742,1487)
	(2724,1510)(2707,1530)(2690,1548)
	(2674,1564)(2657,1578)(2640,1591)
	(2623,1603)(2604,1614)(2585,1623)
	(2565,1632)(2544,1640)(2521,1647)
	(2496,1654)(2469,1659)(2439,1665)
	(2407,1669)(2372,1674)(2336,1677)
	(2298,1681)(2260,1684)(2223,1686)
	(2188,1688)(2158,1690)(2132,1691)
	(2113,1692)(2100,1693)(2093,1693)(2090,1693)
\path(1595,1693)(1592,1693)(1585,1693)
	(1574,1692)(1556,1692)(1533,1691)
	(1505,1689)(1473,1688)(1438,1685)
	(1403,1683)(1369,1680)(1335,1677)
	(1304,1673)(1274,1669)(1247,1665)
	(1222,1660)(1199,1655)(1178,1649)
	(1159,1642)(1141,1635)(1124,1627)
	(1108,1618)(1090,1607)(1074,1595)
	(1058,1581)(1041,1566)(1025,1549)
	(1009,1529)(992,1507)(974,1482)
	(956,1455)(937,1426)(918,1396)
	(899,1365)(882,1335)(866,1307)
	(853,1284)(830,1243)
\path(865.737,1429.341)(830.000,1243.000)(970.394,1370.631)
\path(830,973)(832,970)(837,964)
	(846,954)(860,938)(878,917)
	(899,892)(924,863)(951,833)
	(978,802)(1005,771)(1032,742)
	(1057,715)(1081,690)(1103,667)
	(1124,646)(1143,628)(1161,611)
	(1179,596)(1195,583)(1211,571)
	(1228,560)(1247,549)(1267,538)
	(1287,529)(1308,521)(1330,513)
	(1354,507)(1380,502)(1407,497)
	(1437,492)(1467,488)(1497,485)
	(1525,483)(1550,481)(1570,479)
	(1583,479)(1591,478)(1595,478)
\path(5555,883)(5553,880)(5548,873)
	(5540,860)(5529,842)(5514,819)
	(5496,791)(5476,761)(5455,730)
	(5433,700)(5413,671)(5393,644)
	(5374,619)(5356,597)(5340,577)
	(5324,560)(5308,545)(5293,531)
	(5278,519)(5263,508)(5245,497)
	(5227,488)(5208,479)(5187,472)
	(5165,465)(5140,459)(5113,454)
	(5084,450)(5053,446)(5020,443)
	(4987,440)(4957,437)(4929,436)
	(4908,434)(4893,434)(4884,433)
	(4881,433)(4880,433)
\path(3963.528,824.841)(3845.000,973.000)(3850.926,783.356)
\path(3845,973)(3859,935)(3867,912)
	(3878,884)(3889,853)(3902,820)
	(3915,786)(3929,753)(3942,721)
	(3955,691)(3967,664)(3979,639)
	(3990,616)(4001,596)(4012,578)
	(4022,562)(4033,548)(4044,535)
	(4055,523)(4069,510)(4084,498)
	(4100,488)(4117,479)(4136,472)
	(4157,465)(4180,459)(4206,453)
	(4233,449)(4262,444)(4290,441)
	(4318,438)(4341,436)(4360,435)
	(4374,434)(4381,433)(4385,433)
\path(2874.781,751.735)(2945.000,928.000)(2783.066,829.120)
\path(2945,928)(2918,896)(2902,877)
	(2882,854)(2859,828)(2835,801)
	(2810,773)(2784,745)(2759,718)
	(2736,693)(2713,670)(2692,650)
	(2671,631)(2652,614)(2633,599)
	(2615,586)(2598,574)(2580,563)
	(2563,553)(2543,543)(2522,534)
	(2501,527)(2479,520)(2455,513)
	(2429,508)(2400,503)(2370,498)
	(2337,494)(2302,491)(2266,488)
	(2230,485)(2195,483)(2163,481)
	(2136,480)(2116,479)(2102,478)
	(2094,478)(2090,478)
\path(3800,1198)(3801,1202)(3804,1209)
	(3810,1223)(3818,1243)(3829,1269)
	(3841,1299)(3855,1332)(3870,1366)
	(3885,1400)(3900,1432)(3914,1461)
	(3928,1488)(3941,1513)(3954,1534)
	(3967,1553)(3979,1570)(3992,1585)
	(4004,1599)(4018,1610)(4033,1622)
	(4049,1633)(4066,1642)(4085,1650)
	(4106,1658)(4130,1664)(4156,1669)
	(4184,1674)(4215,1679)(4247,1683)
	(4279,1686)(4309,1688)(4336,1690)
	(4357,1692)(4372,1692)(4381,1693)
	(4384,1693)(4385,1693)
\path(4880,1693)(4881,1693)(4884,1693)
	(4893,1692)(4908,1692)(4930,1690)
	(4957,1688)(4988,1686)(5021,1683)
	(5054,1679)(5086,1674)(5116,1669)
	(5143,1664)(5168,1658)(5191,1650)
	(5213,1642)(5232,1633)(5252,1622)
	(5270,1610)(5285,1600)(5300,1588)
	(5315,1575)(5330,1560)(5346,1543)
	(5363,1525)(5380,1504)(5398,1481)
	(5418,1456)(5439,1429)(5460,1400)
	(5482,1369)(5504,1338)(5525,1308)
	(5545,1279)(5562,1254)(5576,1233)(5600,1198)
\path(5448.721,1312.519)(5600.000,1198.000)(5547.689,1380.383)
\path(3395,1828)(3395,1378)
\path(3335.000,1558.000)(3395.000,1378.000)(3455.000,1558.000)
\put(5780,928){\makebox(0,0)[lb]{\smash{{\SetFigFontNFSS{8}{9.6}{\familydefault}{\mddefault}{\updefault}$l_2$}}}}
\put(3215,973){\makebox(0,0)[lb]{\smash{{\SetFigFontNFSS{8}{9.6}{\familydefault}{\mddefault}{\updefault}$l_0$}}}}
\put(380,973){\makebox(0,0)[lb]{\smash{{\SetFigFontNFSS{8}{9.6}{\familydefault}{\mddefault}{\updefault}$l_1$}}}}
\put(1775,1603){\makebox(0,0)[lb]{\smash{{\SetFigFontNFSS{8}{9.6}{\familydefault}{\mddefault}{\updefault}0}}}}
\put(4565,1648){\makebox(0,0)[lb]{\smash{{\SetFigFontNFSS{8}{9.6}{\familydefault}{\mddefault}{\updefault}0}}}}
\put(4565,388){\makebox(0,0)[lb]{\smash{{\SetFigFontNFSS{8}{9.6}{\familydefault}{\mddefault}{\updefault}1}}}}
\put(1775,388){\makebox(0,0)[lb]{\smash{{\SetFigFontNFSS{8}{9.6}{\familydefault}{\mddefault}{\updefault}2}}}}
\put(4430,73){\makebox(0,0)[lb]{\smash{{\SetFigFontNFSS{8}{9.6}{\familydefault}{\mddefault}{\updefault}$c$}}}}
\put(4430,1963){\makebox(0,0)[lb]{\smash{{\SetFigFontNFSS{8}{9.6}{\familydefault}{\mddefault}{\updefault}$a$}}}}
\put(1595,1963){\makebox(0,0)[lb]{\smash{{\SetFigFontNFSS{8}{9.6}{\familydefault}{\mddefault}{\updefault}$b$}}}}
\put(1685,73){\makebox(0,0)[lb]{\smash{{\SetFigFontNFSS{8}{9.6}{\familydefault}{\mddefault}{\updefault}$d$}}}}
\end{picture}
}}
    \caption{Automaton $\A_3$ with delay blocks.}
    \label{figure:example-three}
  \end{figure}
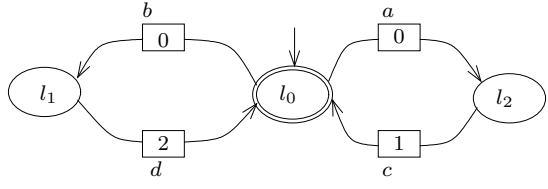
  Consider the ADB $\A_3$  
  in Figure~\ref{figure:example-three} with the initial state
  $l_0$ which is also the only final state.
There are no explicit time advancing $\tick$ transitions in $\A_3$, 
thus, time only advances once the generating run stops at
an accepting state.
Any generating run thus occurs at time 0,
\emph{i.e.} all the transitions are taken at time 0.
Consider a generating run specified by the
transition sequence $r= a\, c\, a\,  c\,  b\,  d\,  a\,  c\,  b\,  d$.
The symbols output at time 0 are the ones corresponding to
delay blocks $\block{0}$, and these are (in order) $aabab$.
The symbols output at time 1 are the ones corresponding to
delay blocks $\block{1}$, and these are (in order) $ccc$.
The symbols output at time 2 are the ones corresponding to
delay blocks $\block{2}$, and these are (in order) $dd$.
Thus, the output timestamped string corresponding to
the transition sequence $r$ is
$\tuple{a, 0}^2 \tuple{b, 0} \tuple{a, 0}  \tuple{b, 0}\, 
 \tuple{c, 1}^3 \, 
  \tuple{d, 2}^2$. 
It can be observed that the 
  the timed language $\lan(\A_3)$ is 
  \[
  \lan(\A_3)=
\left \{ w \left\vert\ \ 
\begin{array}{l}
  w\in (\tuple{a,0} + \tuple{b,0})^* \tuple{c,1}^* \tuple{d,2}^* \text{ and }\\
    |w|_{\tuple{a,0}} = |w|_{\tuple{c,1}} \text{ and }\\
    |w|_{\tuple{b,0}} = |w|_{\tuple{d,2}}
\end{array}
\right.
\right\}
  \]
  where $|w|_{\tuple{\sigma,i}}$ denotes the number of occurrences of
  $\tuple{\sigma,i}$ in $w$.
  The untimed language of $\A_3$ is obtained by projecting the timed words in
the timed language $\lan(\A_3)$  onto the output alphabet.
We have
  \[
  \ulan(\A_3)=
  \set{ u \mid
    u\in (a + b)^* c^* d^* \text{ and }
    |u|_{a} = |u|_{c} \text{ and  }
    |u|_{b} = |u|_{d}};
  \]
  where $|u|_{\sigma}$  denotes the number of occurrences of
  $\sigma$ in $u$.
  \qed
\end{example}
%%length

\smallskip\noindent\textbf{ADB Languages obtained from  the $\oword$ operation on Regular Languages.}
 Given an ADB $\A= (L,D,\Sigma, \delta, l_s, L_f)$ with the output symbol set $\Sigma$ and $M_d$ as the 
largest index for a delay block,
there exists a corresponding regular finite automaton $\reg(\A) =
(L,\Sigma^{\circledcirc}, \delta^{\circledcirc}, l_s, L_f)$ over
 the delay-stamped symbol set $\Sigma^{\circledcirc}= \Sigma\times \set{0, \dots, M_d}\ \cup \set{\tick,\epsilon}$
such that
\begin{compactenum}
\item
$\delta^{\circledcirc}(l, \tuple{\sigma, t}) =  \delta(l,\sigma, \block{t}) $
\item 
$\delta^{\circledcirc}(l,\epsilon) = \delta(l,\epsilon)$.
\item
$\delta^{\circledcirc}(l,\tick) = \delta(l,\tick)$
\end{compactenum}
Intuitively, $\reg(\A)$ is just the ADB $\A$ ``interpreted'' as a regular automaton.
The regular language of  $\reg(\A)$ is denoted by $\rlan(\A)$.
We define  $\oword(\rlan(\A))$ to be the timed word language
$\set{\oword(w) \mid w \in \rlan(\A)}$.

\begin{proposition}
\label{prop:RegularConnection}
Let $\A$ be an ADB, and let  $\reg(\A)$ be the corresponding regular finite automaton with
the corresponding regular language  $\rlan(\A)$.
We have $\lan(\A) =\oword(\rlan(\A))$.
\end{proposition}
%%length
\begin{proof}
By definition.
\qed
\end{proof}
%Sampled timed model. Also finite deterministic 
%interleaving possible because it as an automata. 
%Question: does some problem with arbitrary interleaving becomes undecidable.

\section{Expressiveness of Untimed Languages of ADBs}

In this section we compare the expressive power untimed languages of ADBs
against regular and context free languages.
Given a regular or pushdown automaton $\A$ without timed delay blocks, we let 
$\ulan(\A)$ be the
language of $\A$.
First we show that ADBs can be considered to be a generalization of regular automata.

\begin{proposition}[\tnote{Generalization of regular automata}]
\label{prop:GenReg}
Let $\A$ be a regular automaton without timed delay blocks.
Consider the ADB $\A'$ obtained from $\A$ such that 
(1)~$\A'$ has the same set of
locations,  set of accepting locations and starting
location 
as $\A$; and
(2)~the transition function $\delta^{\A'}$ is such that
$\delta^{\A'}(l,\tuple{\sigma, 0}) = \delta^{\A}(l,\sigma)$ and
$\delta^{\A'}(l,\epsilon) = \delta^{\A}(l,\epsilon)$.
Then, $\ulan(\A) = \ulan(\A')$.
\end{proposition}
\begin{proof}
The ADB $\A'$ has no tick transitions.
Thus since $\A'$ only has delay blocks of duration 0, given a run $r$ of $\A'$, 
the symbols from $\Sigma$ are output in the order in which they are encountered in
the run $r$.
By construction, there is a one to one correspondence between the runs of $\A'$ and 
$\A$ such that the output symbol sequence in a run $r$ of $\A'$ is the same as the
output symbol sequence in the corresponding run of $\A$.
Hence, $\ulan(\A) = \ulan(\A')$.
\qed
\end{proof}

We next show that the expressive power of untimed languages of ADBs is 
incomparable to that of context free languages.

\begin{proposition}
\label{proposition:NestedCFGLan}
Let $\ulan^{\dagger}$ be the untimed language 
\[
\ulan^{\dagger} = \left\{
\begin{array}{c}
a^{n_1} \# a^{n_2} \# \dots \# a^{n_m} \# b^{n_m} \#
 b^{n_{m-1}} \# \dots \#  b^{n_2} \# b^{n_1}\\
 \text{such that}\\
 m \geq 0 \text{ and }
n_i \geq 1 \text{ for } 1\leq i\leq m
\end{array}
\right\}.
\]
There is no ADB $\A$ such that $\ulan(\A) = \ulan^{\dagger}$
\end{proposition}
\begin{proof}
Intuitively, the proof below shows that the maximum delay present in an ADB limits
the ``depth'' of the nestings in the generated word. 

We prove by contradiction.
Let $\A= (L,D,\Sigma, \delta, l_s, L_f)$
 be any ADB with $\set{a,b,\#}$ as the set of output symbols such that
 $\ulan(\A) = \ulan^{\dagger}$.
The automaton $\A$ has a natural graph representation with the nodes in the graph
corresponding to locations and the edges corresponding to $\delta$.
Let $G = \set{S_1, \dots, S_p}$ denote the set of strongly connected 
components (SCCs) of $\A$ which are reachable from $l_s$,
 can reach a final location, and which contain at least one 
$a$-edge.
Observe that every SCC $S_i$ in $G$:
\begin{compactenum}
\item 
  Must have a $b$-edge, and
\item 
  Cannot have a $\tick$-edge (for then an $a$ can be made to appear after a
  $b$).
\item 
  Every $b$-delay in the SCC must be greater than every $a$-delay
  (otherwise an $a$ can be made to appear after a
  $b$ since there are no $\tick$ edges).
\end{compactenum}
% Consider a strings $\widehat{w}= a^{n_1} \# a^{n_2} \# 
% \dots \# a^{n_m} \# b^{n_m} \#
%  b^{n_{m-1}} \# \dots \#  b^{n_2} \# b^{n_1}$ in 
%  $\ulan^{\dagger}$ such that $n_i > |L|$ for all $i$,
% and $m > $.
% For all these strings, 
% the automaton $\A$ must generate $w$ such that $\widehat{w} = \untime(w)$. 
% Except for a finite number of such $\widehat{w}$, 

Consider an SCC $S$ from $G$.
Let $p = l_0 l_1 \dots l_x $ be any path from $l_s$
to an accepting location which passes through $S$.
For $0\leq i\leq j\leq x$, let $p[i]$ denote the state $l_i$ and
$p[i..j]$ the sub-path $l_i\dots l_j$.
 Let
$p[\alpha..\beta]$ be a (maximal) sub- path which lies entirely within $S$, \emph{i.e.}
such that (1)~$\alpha=0$ or  $p[\alpha-1]\notin S$, and
(2)~$\beta=x$ or $p[\beta+1]\notin S$.
Let $C_{\alpha}^1,\dots, C_{\alpha}^q$ be all the non-overlapping cycles (each
$C_{\alpha}^k$ has only one cycle)  in
$ p[\alpha..\beta]$ such that the $C_{\alpha}^j$ occurs after 
 $C_{\alpha}^i$ for $ j >i$.
The sub-path $p[\alpha..\beta]$ can then be written as
\[
\begin{array}{r}
p[\alpha]\, p[\alpha+1] \dots p[e_1] \, {(C_{\alpha}^1)}^{h_1}\,  p[f_1]\dots p[e_2]\,
{(C_{\alpha}^2)}^{h_2} \dots \\
{(C_{\alpha}^q)}^{h_q} \, p[f_q] \dots p[\beta]
\end{array}
\]
for some $e_1, f_1,h_1 \dots e_q, f_q,h_q $ such that the following subpaths
of $ p[\alpha..\beta]$ are all cycle free:
(1)~the path $p[\alpha]\, p[\alpha+1] \dots p[e_1]$;
(2)~the paths $p[f_{k-1}] \dots p[e_k]$  for all $2\leq k\leq q$; and
(3)~the path $p[f_q] \dots p[\beta]$.
 \begin{figure}[t]
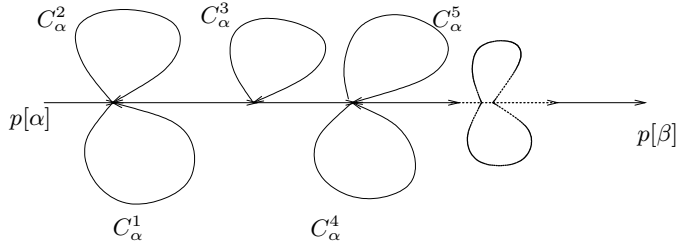

\hspace*{-7mm}
    \strut\centerline{\input Figures/delaycfg.eepic}
    \caption{Structure of sub-path $p[\alpha..\beta]$.}
    \label{figure:DelayCFG}
  \end{figure}
The structure of the subpath $p[\alpha..\beta]$ is illustrated in
Figure~\ref{figure:DelayCFG}.
Thus, $C_{\alpha}^1$  is the first cycle in the $p[\alpha..\beta]$ sub-path,
$C_{\alpha}^2$ is the next cycle %that is different from $C_{\alpha}^1$, 
%$C_{\alpha}^3$ is the next cycle that is different from $C_{\alpha}^2$ (but which
%may be equal to $C_{\alpha}^1$ 
and so on.
Observe that we may have $C_{\alpha}^i = C_{\alpha}^j$ for $i\neq j$.

Observe that each cycle $C_{\alpha}^i $ must have an $a$-transition
(otherwise two consecutive $\#$ symbols can be made to appear in the output).
Note  that for any path $p[\alpha .. \beta]$ and any such subcycle of the path,
we must have that the number of $a$-edges in the subcycle equals the
number of $b$-edges
(otherwise we can ``pump'' the cycle to get $a$'s that are unmatched by $b$'s;
as the matching must occur inside the the same cycle).
Thus, each $C_{\alpha}^i$  can generate some $a^k .. b^k$ pair.
Consider any $ C_{\alpha}^i $.
Let the maximum $a$-delay be $\Delta_a^i$ in $ C_{\alpha}^i $, and let the minimum
$b$-delay be $\Delta_b^i$.
For any $j\neq i$, we must have one of the following to hold:
\begin{compactenum}
\item 
  Either  $\Delta_a^i = \Delta_a^j $ and $\Delta_b^i \neq  \Delta_b^j $; 
\item 
   Or $\Delta_b^i = \Delta_b^j $ and $\Delta_a^i \neq  \Delta_a^j $; or
\item $\Delta_b^i \neq  \Delta_b^j $ and $\Delta_a^i \neq  \Delta_a^j $.
\end{compactenum}
For otherwise, 
if $\Delta_a^i = \Delta_a^j $ and $\Delta_b^i = \Delta_b^j $ 
the two cycles will generate an untimed subpart
$a^{n_{k_1}} \# a^{n_{k_2}} \# b^{n_{k_1}} \# b^{n_{k_2}}$ (with enough
pumping),  when they should
be generating the string with the $b$'s switched (\emph{i.e.} the string
$a^{n_{i_1}} \# a^{n_{i_2}} \# b^{n_{i_2}} \# b^{n_{i_1}}$).
Thus, in any accepting path $p[\alpha..\beta]$ via the SCC $S$,
we can only have at most  $M^2$ subcycles  $C_{\alpha}^1,\dots, C_{\alpha}^q$,
\emph{i.e.} $q \leq M^2$ (as there are only at most $M^2$ distinct values of
$\tuple{\Delta_a,  \Delta_b}$ tuples.

Since the maximum delay is finite (say $M-1$), the cycles in $S$ can generate at most
$M^2$ pairs of unbounded numbers of $a$'s and $b$'s.
That is, the SCC $S$ cannot generate the untimed language
\[
%\mspace{-3mu}
\left\{
%\mspace{-3mu}
\begin{array}{c}
a^{n_1} \# a^{n_2} \# \dots \# a^{n_{M^2+1}} \# b^{n_{M^2+1}} \#
 b^{n_{M^2}} \# \dots \#  b^{n_2} \# b^{n_1} \\
\text{such that}\\
n_i \geq 1 \text{ for } 1\leq i\leq M^2+1
\end{array}%\mspace{-3mu}
\right\}.
\]
Hence, if there are $K$ SCCs in $\A$, then $\A$ cannot generate the untimed
language
\[
%\mspace{-33mu}
\left\{%\mspace{-9mu}
\begin{array}{c}
a^{n_1} \# a^{n_2} \# \dots \# a^{n_{M^2+K+1}} \# b^{n_{M^2+K+1}} \#
 b^{n_{M^2+K}} \# \dots \#  b^{n_2} \# b^{n_1} \\
\text{such that}\\
n_i \geq 1 \text{ for } 1\leq i\leq M^2+K+1 
\end{array}%\mspace{-12mu}
\right\}.
\]
Hence, it follows that there does not exist an ADB $\A$ such that 
$\ulan(\A)=\ulan^{\dagger}$.
\qed
\end{proof}

\begin{proposition}[\tnote{Incomparability with Pushdown Automata}]
\label{prop:Incomparable}
The following assertions hold.
\begin{compactenum}
\item 
There exists an ADB $\A$ such that $\ulan(\A)$ is not context free.

\item 
There exists a visibly pushdown automata $\A$ such that
there is no ADB $\A'$ with $\ulan(\A)= \ulan(\A')$. 
\end{compactenum}
\end{proposition}
\begin{proof}
For the first part of the proposition, consider the ADB $\A_1$ of 
Figure~\ref{figure:example-one}.
The untimed language of $\A_1$ is $\set{a^n b^n c^n \mid n\geq 0}$ which
is not context free.
\\
The second part of the theorem follows from Proposition~\ref{proposition:NestedCFGLan},
noting that there exists a visibly pushdown automaton which generates the
language $\ulan^{\dagger}$.
\qed
\end{proof}

Proposition~\ref{prop:Incomparable} shows that there is a tradeoff between the
expressive power of ADBs and pushdown automata. 
On one hand, ADBs are not restricted to matching only once (i.e., they can generate
$a^n b^n c^n$), but on the other they lose the infinite nesting
capability of pushdown automata (e.g., in the language $\ulan^{\dagger}$ of
Proposition~\ref{proposition:NestedCFGLan}).

\begin{theorem}[\tnote{Expressive power of ADBs}]
The following assertions hold: 
(1)~The class of untimed languages of ADBs strictly subsumes the class of 
regular languages.
(2)~The class of untimed languages of ADBs is incomparable in expressive power
as compared to the class of context-free languages.
\end{theorem}
\begin{proof}
The results follow from Propositions~\ref{prop:GenReg} and~\ref{prop:Incomparable}.
\qed
\end{proof}

%%\mynote{Open}\\\medskip\noindent{\bf Comparison with context sensitive languages.} This 
%%is still a question. But nice and simple model to express $a^n b^n c^n$.

%%\medskip\noindent{\bf Subset of languages of a decidable languages.}

%%% Local Variables: 
%%% mode: latex
%%% TeX-master: t
%%% End: 

\section{Closure Properties} %% of ADB Languages}
In this section we will study the closure properties of 
timed and untimed languages of ADBs with respect to 
operations like union, intersection, complement, concatenation
and Kleene star.

\begin{proposition}[\tnote{Closure under union}]
\label{prop:UnionClosure}
Let $\A_1$ and $\A_2$  be  ADBs.
There exists an ADB $\A$ such that $\lan(\A) = \lan(\A_1) \cup \lan(A_2)$.
\end{proposition}
\begin{proof}
The ADB $\A$ has a special initial states, and two $\epsilon$ transitions from this
initial state to copies of $\A_1$ and $\A_2$.
\qed
\end{proof}

\begin{proposition}[\tnote{Closure under intersection with regular languages}]
\label{prop:IntReg}
Given an untimed ADB language $\ulan$ and a regular language $\rlan$, the 
language $\ulan \cap \rlan$ is an untimed ADB language.
\end{proposition}
\begin{proof}
Given an ADB $\A_1$ with $\ulan$ and a finite-state automata $\A_2$ 
for a regular language $\rlan$, we will present an explicit construction 
of an ADB with untimed language $\ulan \cap \rlan$ 
in Proposition~\ref{prop:IntRegEmpty}.
The desired result will follow from the construction.
\qed
\end{proof}

\smallskip\noindent{\bf Concatenation and Kleene star.}
We will now consider closure under concatenation and Kleene star.
Given untimed languages $\ulan$, $\ulan_1$ and $\ulan_2$, we define their
concatenation $\ulan_1 \circ \ulan_2$ and Kleene star $\ulan^*$ as follows:
\begin{align*}
\ulan_1 \circ \ulan_2 & \triangleq 
\set{ w_1\circ w_2 \mid w_1 \in \ulan_1 \text{ and } w_2 \in \ulan_2}\\
\ulan^* & \triangleq 
\set{ \epsilon } \cup \set{ w_1 \circ w_2 \circ \ldots \circ w_i \mid i \in \mathbb{N}, 1 \leq j \leq i. 
w_j \in \ulan}
\end{align*}

\begin{proposition}[\tnote{Closure  under concatenation}]
\label{prop:Concat}
Let $\ulan_1$ and $\ulan_2$ be untimed ADB languages.
Then $\ulan_1 \circ \ulan_2$ is an untimed ADB language.
\end{proposition}
\begin{proof}
Let $\ulan_1 = \ulan(\A_1)$ and $\ulan_2 = \ulan(\A_2)$.
Let the largest delay in the delay blocks of $\A_1$ be $M$.
Consider the ADB $\A$ in Figure~\ref{figure:Concat}.
The ADB $\A$ has $M$ occurrences of $\tick$ transitions from
every final location of $\A_1$ to the starting location of $\A_2$.
%
\begin{comment}
\begin{figure}[!h]
  \hspace*{5mm}
  \begin{minipage}[t]{0.85\linewidth}
    \strut\centerline{\input Figures/concat.eepic}
    \caption{Automatons $\A_1$ and $\A_2$ concatenated to get   $\ulan_1 \circ \ulan_2$.}
    \label{figure:Concat}
  \end{minipage}\\
  \begin{minipage}[b]{0.85\linewidth}
    \vspace*{7mm}
    \strut\centerline{\input Figures/star.eepic}
    \caption{Automatons $\A_*$ using $\A$ to get   $\ulan^*$.}
    \label{figure:Kleene}
  \end{minipage}
\end{figure}
\end{comment}
% 
%
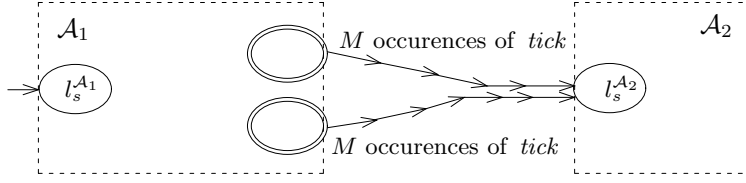
\begin{figure}[!h]
  \hspace*{-5mm}
 % \begin{minipage}[t]{0.85\linewidth}
    \strut\centerline{\setlength{\unitlength}{0.00043745in}
\begingroup\makeatletter\ifx\SetFigFont\undefined%
\gdef\SetFigFont#1#2#3#4#5{%
  \reset@font\fontsize{#1}{#2pt}%
  \fontfamily{#3}\fontseries{#4}\fontshape{#5}%
  \selectfont}%
\fi\endgroup%
{\renewcommand{\dashlinestretch}{30}
\begin{picture}(8844,2064)(0,-10)
\put(3319,575){\ellipse{854}{584}}
\put(3319,575){\ellipse{944}{674}}
\put(3319,1430){\ellipse{854}{584}}
\put(3319,1430){\ellipse{944}{674}}
\put(7132,1023){\ellipse{854}{584}}
\put(809,1009){\ellipse{854}{584}}
\dashline{60.000}(8832,12)(8832,2037)(6717,2037)
	(6717,12)(8832,12)
\path(3792,1452)(4422,1317)
\path(4233.424,1296.047)(4422.000,1317.000)(4258.567,1413.383)
\path(4422,1317)(5097,1182)
\path(4908.729,1158.466)(5097.000,1182.000)(4932.262,1276.136)
\path(5052,1182)(5682,1047)
\path(5493.424,1026.047)(5682.000,1047.000)(5518.567,1143.383)
\path(5682,1047)(6132,1047)
\path(5952.000,987.000)(6132.000,1047.000)(5952.000,1107.000)
\path(6132,1047)(6717,1047)
\path(6537.000,987.000)(6717.000,1047.000)(6537.000,1107.000)
\path(3792,552)(4377,642)
\path(4208.217,555.327)(4377.000,642.000)(4189.970,673.932)
\path(4332,642)(4962,777)
\path(4798.567,680.617)(4962.000,777.000)(4773.424,797.953)
\path(4962,777)(5412,912)
\path(5256.832,802.808)(5412.000,912.000)(5222.350,917.747)
\path(5412,912)(5862,912)
\path(5682.000,852.000)(5862.000,912.000)(5682.000,972.000)
\path(5817,912)(6222,912)
\path(6042.000,852.000)(6222.000,912.000)(6042.000,972.000)
\path(6222,912)(6717,912)
\path(6537.000,852.000)(6717.000,912.000)(6537.000,972.000)
\dashline{60.000}(3747,12)(372,12)(372,2037)
	(3747,2037)(3747,12)
\path(12,1002)(372,1002)
\path(192.000,942.000)(372.000,1002.000)(192.000,1062.000)
\put(597,1632){\makebox(0,0)[lb]{\smash{{\SetFigFont{10}{12.0}{\familydefault}{\mddefault}{\updefault}$\A_1$}}}}
\put(8202,1677){\makebox(0,0)[lb]{\smash{{\SetFigFont{10}{12.0}{\familydefault}{\mddefault}{\updefault}$\A_2$}}}}
\put(7077,912){\makebox(0,0)[lb]{\smash{{\SetFigFont{9}{10.8}{\familydefault}{\mddefault}{\updefault}$l^{\A_2}_s$}}}}
\put(687,912){\makebox(0,0)[lb]{\smash{{\SetFigFont{9}{10.8}{\familydefault}{\mddefault}{\updefault}$l^{\A_1}_s$}}}}
\put(3927,1497){\makebox(0,0)[lb]{\smash{{\SetFigFont{9}{10.8}{\familydefault}{\mddefault}{\updefault}$M \text{ occurences of } \tick$}}}}
\put(3837,237){\makebox(0,0)[lb]{\smash{{\SetFigFont{9}{10.8}{\familydefault}{\mddefault}{\updefault}$M \text{ occurences of } \tick$}}}}
\end{picture}
}}
    \caption{Automatons $\A_1$ and $\A_2$ concatenated to get   $\ulan_1 \circ \ulan_2$.}
    \label{figure:Concat}
  \end{figure}
  Let $r_1$ and $r_2$ be accepting runs in $\A_1$ and $\A_2$
which result in the untimed  output words $w_1$ and $w_2$
respectively.
Let $r_1 = l_0 \stackrel{\alpha_0}{\longrightarrow} l_1  
\stackrel{\alpha_1}{\longrightarrow}\dots l_n$.
Then, $r_1^{\dagger} = l_0 \stackrel{\alpha_0}{\longrightarrow} l_1  
\stackrel{\alpha_1}{\longrightarrow}\dots l_n 
\underset{ M \text{ occurences of } \tick}
{\underbrace{\stackrel{\tick}\longrightarrow l_{n_1}
 \stackrel{\tick}\longrightarrow l_{n_2} \dots l_{n_M}}}$ 
forms a part of an accepting run in $\A$,
namely the run $r_1^{\dagger}  \circ \stackrel{\epsilon}{\rightarrow} \circ r_2$.
Moreover, the run $r_1^{\dagger}$ generates the untimed output word
$w_1 \circ w_2$, as the maximum delay of a delay block in $r_1$ is $M$, thus
all the output symbols have been output before  the $r_2$ part in
$r_1^{\dagger}$ starts.

In the other direction, given an accepting  run $r_1^{\dagger}$ of $\A$,
it can be decomposed  into $r_1 \circ \underset{ M \text{ occurences of } \tick}
{\underbrace{\stackrel{\tick}\longrightarrow l_{n_1}
 \stackrel{\tick}\longrightarrow l_{n_2} \dots l_{n_M}}}\,  \circ\, r_2$ such that
$r_1$ and $r_2$ are accepting runs in $\A_1$ and $\A_2$ respectively.
The result follows from noting that $ r_1^{\dagger}$ generates the
untimed output word $w_1 \circ w_2$ where
$w_1$ and $w_2$ are untimed output words generated by the runs
$r_1$ and $r_2$ respectively.
\qed
\end{proof}
% \begin{proof}
% The proof can be found in the appendix.
% \qed
% \end{proof}

\begin{proposition}[\tnote{Closure under Kleene star}]
\label{prop:Kleene}
Let $\ulan$ be an untimed ADB language.
Then $\ulan^*$ is an untimed ADB language.
\end{proposition}
%%%length
\begin{proof}
Let $\ulan = \ulan(A)$.
Consider the ADB $\A_*$ in Figure~\ref{figure:Kleene}.
\begin{figure}[!h]
  \hspace*{4mm}
 %   \vspace*{7mm}
    \strut\centerline{\setlength{\unitlength}{0.00043745in}
\begingroup\makeatletter\ifx\SetFigFont\undefined%
\gdef\SetFigFont#1#2#3#4#5{%
  \reset@font\fontsize{#1}{#2pt}%
  \fontfamily{#3}\fontseries{#4}\fontshape{#5}%
  \selectfont}%
\fi\endgroup%
{\renewcommand{\dashlinestretch}{30}
\begin{picture}(9190,2958)(0,-10)
\put(3620,1469){\ellipse{854}{584}}
\put(3620,1469){\ellipse{944}{674}}
\put(3620,2324){\ellipse{854}{584}}
\put(3620,2324){\ellipse{944}{674}}
\put(7445,1874){\ellipse{854}{584}}
\put(7445,1874){\ellipse{944}{674}}
\put(1110,1903){\ellipse{854}{584}}
\path(4093,2346)(4723,2211)
\path(4534.424,2190.047)(4723.000,2211.000)(4559.567,2307.383)
\path(4723,2211)(5398,2076)
\path(5209.729,2052.466)(5398.000,2076.000)(5233.262,2170.136)
\path(5353,2076)(5983,1941)
\path(5794.424,1920.047)(5983.000,1941.000)(5819.567,2037.383)
\path(5983,1941)(6433,1941)
\path(6253.000,1881.000)(6433.000,1941.000)(6253.000,2001.000)
\path(6433,1941)(7018,1941)
\path(6838.000,1881.000)(7018.000,1941.000)(6838.000,2001.000)
\path(4093,1446)(4678,1536)
\path(4509.217,1449.327)(4678.000,1536.000)(4490.970,1567.932)
\path(4633,1536)(5263,1671)
\path(5099.567,1574.617)(5263.000,1671.000)(5074.424,1691.953)
\path(5263,1671)(5713,1806)
\path(5557.832,1696.808)(5713.000,1806.000)(5523.350,1811.747)
\path(5713,1806)(6163,1806)
\path(5983.000,1746.000)(6163.000,1806.000)(5983.000,1866.000)
\path(6118,1806)(6523,1806)
\path(6343.000,1746.000)(6523.000,1806.000)(6343.000,1866.000)
\path(6523,1806)(7018,1806)
\path(6838.000,1746.000)(7018.000,1806.000)(6838.000,1866.000)
\path(9178,951)(7828,1671)
\path(8015.059,1639.235)(7828.000,1671.000)(7958.588,1533.353)
\path(808,2616)(1033,2211)
\path(893.135,2339.210)(1033.000,2211.000)(998.034,2397.487)
\dashline{60.000}(4048,906)(673,906)(673,2931)
	(4048,2931)(4048,906)
\path(7423,1536)(7422,1534)(7420,1529)
	(7417,1521)(7412,1508)(7404,1490)
	(7394,1467)(7382,1439)(7368,1406)
	(7352,1370)(7335,1330)(7316,1289)
	(7296,1247)(7275,1205)(7255,1163)
	(7234,1123)(7213,1084)(7192,1047)
	(7171,1012)(7151,979)(7130,949)
	(7110,920)(7089,894)(7068,869)
	(7046,846)(7024,824)(7002,804)
	(6978,784)(6954,766)(6928,748)
	(6908,736)(6887,723)(6866,711)
	(6844,700)(6821,688)(6796,677)
	(6771,666)(6745,655)(6717,644)
	(6688,634)(6658,624)(6626,614)
	(6594,604)(6559,595)(6524,586)
	(6487,577)(6449,568)(6409,560)
	(6368,552)(6326,544)(6283,537)
	(6239,530)(6193,523)(6147,516)
	(6099,510)(6050,504)(6001,498)
	(5951,493)(5900,488)(5848,483)
	(5795,479)(5742,475)(5688,471)
	(5633,467)(5578,463)(5521,460)
	(5464,457)(5406,454)(5346,451)
	(5286,448)(5240,447)(5193,445)
	(5145,443)(5097,441)(5047,440)
	(4997,438)(4946,437)(4893,435)
	(4840,434)(4785,432)(4730,431)
	(4673,430)(4615,429)(4557,428)
	(4498,426)(4437,425)(4376,424)
	(4314,423)(4252,423)(4188,422)
	(4124,421)(4060,420)(3995,419)
	(3930,419)(3864,418)(3799,418)
	(3733,417)(3667,417)(3602,416)
	(3536,416)(3471,416)(3406,416)
	(3342,415)(3278,415)(3214,415)
	(3152,415)(3090,415)(3029,415)
	(2968,415)(2909,415)(2851,415)
	(2793,415)(2736,415)(2681,416)
	(2626,416)(2573,416)(2520,416)
	(2469,416)(2419,417)(2369,417)
	(2321,417)(2273,418)(2226,418)
	(2181,418)(2120,419)(2060,420)
	(2002,420)(1945,421)(1888,422)
	(1833,423)(1778,424)(1724,425)
	(1670,426)(1618,428)(1566,430)
	(1515,432)(1464,434)(1414,436)
	(1366,439)(1318,442)(1271,445)
	(1225,449)(1180,453)(1137,457)
	(1094,462)(1053,467)(1013,472)
	(974,478)(936,484)(900,491)
	(865,497)(831,505)(799,512)
	(768,520)(738,528)(709,537)
	(682,546)(655,555)(630,565)
	(605,576)(582,586)(559,597)
	(537,609)(516,621)(488,638)
	(461,656)(435,675)(409,695)
	(384,717)(360,739)(336,763)
	(313,788)(290,814)(268,841)
	(247,869)(227,898)(207,927)
	(188,957)(170,988)(154,1019)
	(138,1050)(123,1082)(109,1113)
	(96,1144)(84,1175)(73,1206)
	(64,1236)(55,1265)(47,1294)
	(40,1323)(34,1351)(29,1378)
	(24,1405)(21,1431)(17,1461)
	(14,1491)(13,1520)(12,1549)
	(12,1578)(12,1607)(14,1636)
	(17,1664)(20,1692)(25,1720)
	(30,1747)(36,1773)(43,1798)
	(51,1823)(60,1846)(69,1868)
	(79,1888)(89,1907)(100,1925)
	(111,1941)(122,1956)(134,1969)
	(147,1981)(159,1992)(172,2001)
	(186,2008)(202,2016)(219,2022)
	(237,2025)(256,2027)(277,2027)
	(300,2025)(325,2021)(352,2015)
	(382,2007)(413,1998)(447,1987)
	(481,1974)(516,1961)(551,1948)
	(583,1935)(611,1923)(634,1913)(673,1896)
\path(484.020,1912.924)(673.000,1896.000)(531.970,2022.927)
\put(4228,2391){\makebox(0,0)[lb]{\smash{{\SetFigFont{9}{10.8}{\familydefault}{\mddefault}{\updefault}$M \text{ occurences of } \tick$}}}}
\put(4138,1131){\makebox(0,0)[lb]{\smash{{\SetFigFont{9}{10.8}{\familydefault}{\mddefault}{\updefault}$M \text{ occurences of } \tick$}}}}
\put(7378,1806){\makebox(0,0)[lb]{\smash{{\SetFigFont{9}{10.8}{\familydefault}{\mddefault}{\updefault}$l^{\A_*}_s$}}}}
\put(898,2526){\makebox(0,0)[lb]{\smash{{\SetFigFont{10}{12.0}{\familydefault}{\mddefault}{\updefault}$\A$}}}}
\put(988,1806){\makebox(0,0)[lb]{\smash{{\SetFigFont{9}{10.8}{\familydefault}{\mddefault}{\updefault}$l^{\A}_s$}}}}
\put(2923,96){\makebox(0,0)[lb]{\smash{{\SetFigFont{9}{10.8}{\familydefault}{\mddefault}{\updefault}$\epsilon$}}}}
\end{picture}
}}
    \caption{Automatons $\A_*$ using $\A$ to get   $\ulan^*$.}
    \label{figure:Kleene}
\end{figure}
$l_s^{\A}$ is the starting location of $\A$, and $l_s^{\A_*}$ is the
starting location of $\A_*$.
The ADB $\A_*$ has $M$ occurrences of $\tick$ transitions from
every final location of $\A$ to the starting location of $\A_*$.
There is also an $\epsilon$ transition from the starting location of $\A_*$
to the starting location of $\A$.
% %
% \begin{figure}
%   \strut\centerline{\picfile{star.eepic}}
%   \caption{Automatons $\A_*$ using $\A$ to get   $\ulan^*$.}
%   \label{figure:Kleene}
% \end{figure}
%
%
It can be seen that  the untimed output language of $\A_*$ is $\ulan^*$.
The proof follows along similar lines to the  proof of
Proposition~\ref{prop:Concat}.
\qed
\end{proof}
% \begin{proof}
% The proof can be found in the appendix.
% \qed
% \end{proof}
%%length

We will now show the ADB languages are not closed under some Boolean 
operations, and towards this goal we first prove a pumping lemma.

\begin{proposition}[\tnote{Pumping Lemma for ADB runs}]
\label{prop:PumpingLemmaRuns}
Let $\A$ be an ADB and let $L$ be the set of locations of $\A$.
Let $w \in \lan(\A)$  with $|w| > |L|$ be the output timed word corresponding to an
accepting run $r= l_0 \stackrel{\alpha_0}{\longrightarrow} l_1  
\stackrel{\alpha_1}{\longrightarrow}\dots l_n$.
Consider any subrun  $r_s$ of $r$, \emph{i.e.}
$r= r_0\circ r_s\circ r_1$, such that $r_s$ contains at least 
$|L|$ transitions.
%$l_k \stackrel{\alpha_k}{\longrightarrow} l_{k+1}  
%\stackrel{\alpha_{k+1}}{\longrightarrow}\dots l_{j+1}$ 
Then, there exists a subrun $r_p$ of $r_s$, \emph{i.e.}
$r_s= r_{s_0} \circ r_p \circ  r_{s_1}$
 with $r_p$ containing
at most $|L|$ transitions such that for all $i\geq 0$
the runs $r_0 \circ  r_{s_0} \circ  \rep_i(r_p) \circ r_{s_1} \circ  r_1$
are also accepting runs of $\A$.
%$\oword \left (\alpha_0 \alpha_1 \dots \alpha_{k-1} \rep_i\left( \alpha_k \alpha_{k+1}
%\dots \alpha_j \right) \alpha_{j+1} \dots \alpha_n \right)$  for all $i\geq 0$ also
%belong to  $\lan(\A)$.

% For every  timed word $w\in \lan(\A)$ such that $|w| \geq |L|$,
% there exist non-overlapping substrings $\overline{w}_0,\dots,\overline{w}_m$ of $w$
% with $m\geq 0$ and $0 < \sum_{j=0}^m |\overline{w}_j| \, \leq  |L|  $
% and $w = \overline{w}_0^p \circ \overline{w}_0 \circ \dots \circ  \overline{w}_m^p \circ 
% \overline{w}_m \circ \overline{w}^s$  such that
% for every $i \geq 0$, we have that
% the string 
% $\overline{w}_0^p \circ \rep_i(\overline{w}_0) \circ \dots  \overline{w}_m^p \circ \rep_i(\overline{w}_m) 
% \circ \overline{w}^s$ is 
% also in $\lan(\A)$. 
\end{proposition}
\begin{proof}
The proof follows from the pumping lemma for regular finite state automata, and
from Proposition~\ref{prop:RegularConnection}.
\qed
\end{proof}

\begin{remark}
There are difficulties in obtaining a pumping lemma for timed words.
We give an example.
Let $r, r_0, r_{s_0}, r_p, r_{s_1}, r_1$ be as in 
Proposition~\ref{prop:PumpingLemmaRuns}.
Let $w$ be the timed word corresponding to the run $r$.
Let $\tuple{\sigma_{\alpha}, t_{\alpha}}$ and  $\tuple{\sigma_{\beta}, t_{\beta}}$
be the timestamped
symbols generated by some transition in $ r_0 \circ r_{s_0}$, and 
by some transition in  $r_{s_1}\circ  r_1$ respectively.
Let us denote these two transitions as $\tran_{\alpha}$ and
 $\tran_{\beta}$.
We may have 
$t_{\alpha} >  t_{\beta}$, \emph{i.e.} $\tuple{\sigma_{\alpha}, t_{\alpha}}$ 
appears after $\tuple{\sigma_{\beta}, t_{\beta}}$ in the timed word $w$,
even though the transition which generates $\tuple{\sigma_{\alpha}, t_{\alpha}}$ 
occurs before the transition which generates $\tuple{\sigma_{\beta}, t_{\beta}}$.
Let the number of $\tick$ transitions in $r_p$ be $\Delta_p$.
Each ``pump'' of $r_p$ introduces an additional delay of $\Delta_p$between  when
the transitions  $\tran_{\alpha}$ and  $\tran_{\beta}$ occur.
Eventually, after enough pumps, the delay will large enough
that the timestamped output symbol corresponding to   $\tran_{\beta}$ 
will appear after the timestamped output symbol corresponding to
$\tran_{\alpha}$.
Thus, when we pump an accepting run, the resulting timed word,
with each pump,
may undergo
 a \emph{reordering} of the output symbols corresponding
to the unpumped run parts.
There is also a reordering corresponding to the pumped run part.
\qed
\end{remark}

\begin{proposition}[\tnote{Non-closure under intersection}]
\label{prop:NonClosureInt}
There exist ADBs $\A_1$ and $\A_2$ such that 
(1)~$\lan(\A_1) \cap \lan(A_2)$ is
not an ADB language, and
(2)~$\ulan(\A_1) \cap \ulan(A_2)$ is
not an untimed ADB language. 
\end{proposition}
\begin{proof}
(Sketch.)
  Consider the language
  \[
  \lan^{\dagger} = \left\{\kappa_0(w \#) \kappa_1(w \#)
    \kappa_2(w \#)  \dots \kappa_n(w \#) \left\vert 
     \begin{array}{l} 
      w\in\set{a,b}^* \text{ and }\\
      n\geq 0
      \end{array}
      \right. % \mspace{-16mu}
 \right\}
  \]
 where $\kappa_i()$ is the function defined in Section~\ref{section:Model}.
  We show $\lan^{\dagger}$ is not an ADB language, and that
  there exist ADBs $\A_1$ and $\A_2$ such that
  $\lan^{\dagger} = \lan(\A_1) \cap   \lan(A_2)$.
To show the first claim, let 
  $ \lan^{\dagger}$ be the output
  language of an ADB $\A^{\dagger}$ containing $K$ locations.
  Consider a timed word $w_{\dagger} =
  \kappa_0(w \#) \kappa_1(w \#)
  \kappa_2(w \#)  \dots \kappa_{K+2}(w \#) $ with
  $|w| > K$.
  Let $r_{\dagger}$ be the generating run for $w_{\dagger}$.
  Using the pumping lemma,  we can show there exists  a subrun $r_p$ of
  $r_{\dagger}$ such that 
  (1)~the subrun $r_p$ contains at least one output symbol
  transition, and
  (2)~the subrun contains at most $K$ output symbol transitions; and
  (3)~for $r_{\dagger} = r_0\circ r_p \circ r_1$, we have that
  $r_0 \circ r_1$ is also a generating run for $\A^{\dagger}$ (\emph{i.e.},
  we pump down $r_p$.
  Let $w_{01}$ be the output word corresponding to
  the generating run  $r_0 \circ r_1$.
  Because of the constraints on $r_p$, we have that
  $w_{01}$ contains at least one, and at most $K$ output symbols
  less than $w$.
  It can be checked that this means that $w_{01}$ is not a member of $\lan^{\dagger} $,
  a contradiction. 
  To show that $\lan^{\dagger}$ is the intersection of two ADB languages,
  we consider two ADBs, the first ADB checks that the word with
  timestamp $2j$ matches the word with time $2j+1$ for all $j$;
  the second ADB checks that the word with
  timestamp $2j+1$ matches the word with time $2j+2$ for all $j$.
  It can checked that such ADBs exist and that the intersection of the languages
  is $\lan^{\dagger}$.
 % The detailed proof is given in the Appendix.
   \qed
\end{proof}

\begin{proposition}[\tnote{Non-closure under complementation}]
\label{prop:NonClosureComp}
Given and ADB $\A$, let $\overline{\lan}(\A)$ denote the  complement language of
$\lan(\A)$.
There exists an ADB  $\A$ such that for all ADBs $\A'$ we have $\overline{\lan}(\A)  \neq
\lan(\A')$, and $\overline{\ulan}(\A)  \neq \ulan(\A')$.
\qed
\end{proposition}
% \begin{proof}
% The proof can be found in the appendix.
% \qed
% \end{proof}

We summarize our results in the following theorem.
\begin{theorem}[\tnote{Closure properties}]
The class of untimed languages of ADBs are closed
under union, concatenation, Kleene star, and intersection with 
regular languages, but not closed under intersection and complementation.
\qed
\end{theorem}
%\begin{proof}
%%The results follow from Propositions~\ref{prop:UnionClosure},
%%~\ref{prop:NonClosureInt} and~\ref{prop:NonClosureComp}.
%\qed
%\end{proof}

\section{Decision Problems and Model Checking}
In this section we first study the decision problems such as
emptiness, universality for ADBs, and then study the model 
checking problem.
In the model checking problem we consider an ADB as the model
to generate words, and a specification given as regular language.
Our goal is to check the containment of the untimed language
of the ADB in the regular language.

\subsection{Decision Problems}

\begin{proposition}[\tnote{Emptiness checking}] 
\label{proposition:Emptiness}
Given an ADB $\A$, it can be checked in linear time whether $\lan(\A) =\emptyset$.
\end{proposition}
\begin{proof}
The proposition follows from the fact that $\lan(\A)$ is non-empty iff there is a
path from the initial location to an accepting location.
\qed
\end{proof}

\begin{proposition}[\tnote{Membership checking of timed words}]
\label{prop:Membership}
Given an ADB $\A$ with $n_\A$ locations and $m_\A$ edges, and a timed word $w$, checking
whether $w \in \lan(\A)$ can be checked in time $O\left(M \cdot \left(n_{\A} + m_{\A} + |w| +  \tend\right)\right)$,
where 
$\tend$ is the largest timestamp in $w$, and $M$ is the largest delay of a delay block
in $\A$
(thus, if $M$ is a constant and  $\tend = O( n_{\A} + m_{\A} + |w| )$, 
then we have a linear  time algorithm).
\end{proposition}
\begin{proof}
(Sketch.)
Let $w = \tuple{w_0^{\sigma}, w_0^t} \tuple{w_1^{\sigma}, w_1^t} \dots
\tuple{w_n^{\sigma}, w_n^t}$.
Let the end timestamp of $w$ be $\tend$ (\emph{i.e.}  $   w_n^t = \tend$).
We first construct a finite state deterministic  regular automaton $\A_w$ with just one path 
(corresponding to $w$) over the
alphabet 
$\Sigma_w = \set{\tick, \tuple{w_0^{\sigma}, w_0^t},  \tuple{w_1^{\sigma}, w_1^t}
\dots \tuple{w_n^{\sigma}, w_n^t}}$.
We then construct a (non-deterministic)  ADB  $\A^{\ddagger}$ based on
$\A$ and $\A_w$ such that
$\A^{\ddagger}$ has an accepting  path  iff
$\A$ outputs the timed word $w$.
Let $M$ be the largest delay of a delay block in $\A$.
The automaton $\A^{\ddagger}$ will simulate the executions of $\A$; and of $\A_w$ 
simultaneously for
the current time, and for time upto $M$ time units in the future.
That is, the automaton $\A^{\ddagger}$ is able to verify that $\A_w$ 
 first generates the output symbols corresponding to the current time outputs in $\A$,
then generate output symbols corresponding to current time plus one in $\A$, and so on.
The details of the construction are omitted for lack of space.
%The details of the construction can be found in the appendix.
\qed
\end{proof}

% \begin{proposition}
% \label{proposition:IntersectionEmptiness}
% Given three arbitrary ADBs $\A_1$, $\A_2$ and $\A_3$, it is undecidable to check
% whether $\lan(\A_1) \cap \lan(\A_2) \cap \lan(\A_3)= \emptyset$.
% \end{proposition}
% \begin{proof} 
%  The proof can be found in the Appendix.
%   \qed
% \end{proof}

\begin{proposition}[\tnote{Universality}]
\label{proposition:Universality}
Let $\A$ be an ADB.
It is undecidable to check whether $\ulan(\A) = \Sigma^*$.
\end{proposition}
\begin{proof}
Let $\T$ be a Turing machine with $\Sigma$ as the tape alphabet 
%and 
%$F$ as the set of accepting locations.
Let a valid computation of $\T$ be denoted by an untimed string
$w= w_0 \# w_1 \#\dots w_n$ such that $w_0$ represents the initial tape
configuration, $w_{i+1}$ is a tape configuration that follows from $w_i$ for
$i\geq 0$, and $\#$ is a special delimiter symbol.
%Let $\tilde{w}$ be a timed word corresponding to $w$ such that
%$w_i \#$ occurs at time $i$ for $i\geq 0$.
We first show that there exists an ADB $\A$ such that
$\ulan(\A) $ is the set of strings denoting the invalid computations
of $\T$.

If a string $w$ represents an invalid computation, then one
of the following conditions must hold.
\begin{compactenum}
\item
  The string $w$ is not of the form  $w_0 \# w_1 \#\dots w_n$, where each $w_i$
  denotes a tape configuration.
\item
  $w_0$ is not an initial tape configuration.
\item
  $w_n$ is not an accepting  tape configuration.
\item
  $w_{i+1}$ does not follow from $w_i$ for some $i$.
  
\end{compactenum}

The set of strings satisfying conditions 1,2 or 3 is regular.
There exists an ADB $\A'$ such that $\ulan(\A')$ is the set of strings
satisfying the last condition.
The automaton $\A'$ first generates strings from $(\Sigma\cup\set{\#})^*$ at
time 0.
It then non-deterministically moves to a location from which it generates
$w_i \# w_{i+1}$ such that 
(a)~$w_i$ is a configuration (\emph{i.e.}, a string from $\Sigma^* Q \Sigma^*$ where
$Q$ is the set of locations of $\T$,
(b)~the configuration $w_i$  is generated at time 0 and $w_{i+1}$ generated
at time 1, and
(c)~$w_{i+1}$ is not a configuration that follows from $w_i$ 
(this can be done by ``knowing'' some  future two symbols of $w_i$ at time 0, and
accordingly generating a symbol at time 1 for $w_{i+1}$ such that 
$w_{i+1}$ cannot be a configuration following $w_i$.
Once such $w_i \# w_{i+1}$ is generated, $\A'$ then generates strings
from $(\Sigma\cup\set{\#})^*$ at time 2.
Since ADBs are closed under union, we can take the union of ADBs generating untimed
strings satisfying either of the four conditions.
The ADBs for conditions 1,2 and 3 ``operate'' at times
3, 4 and 5 respectively (\emph{i.e.} they are regular automatons with delay blocks
of 3,4 and 5 respectively at every transition).
This union ADB $\A$ will then generate all untimed strings denoting  invalid computations
of $\T$.

Now, if were decidable to check whether $\ulan(\A) = \Sigma^*$, then it would mean
we can check whether the language of $\T$ is non-empty, as the language of $\T$ is
non-empty iff the there exists a valid computation of $\T$, and 
a valid computation of $\T$ exists iff $\ulan(\A) \neq \Sigma^*$,
Thus, if $\A$ is an ADB, it is undecidable in general to check whether
$\ulan(\A) = \Sigma^*$.
\qed
\end{proof}

\begin{corollary}[\tnote{Equivalence to a regular language}]
\label{cor:EquivReg}
Let $\A$ be an ADB and let $\R$ be a regular language.
It is undecidable to check whether $\ulan(\A) = \R$.
\qed
\end{corollary}

\begin{corollary}[\tnote{Containment in another delay model}]
\label{cor:Contain}
Let $\A$ and $\A'$ be  ADBs.
It is undecidable to check whether $\ulan(\A) \subseteq \ulan(\A')$.
\end{corollary}
\begin{proof}
We reduce universality of untimed languages of ADBs to this problem.
Let $\A$ be an ADB such that $\ulan(\A) = \Sigma^*$.
Then, given any ADB $\A'$, we have $\ulan(\A) \subseteq \ulan(\A')$ iff
$\ulan(\A') = \Sigma^*$.
\qed
\end{proof}

\begin{theorem}[\tnote{Decision problems}]
The following assertions hold for timed and untimed ADB languages:
\begin{compactenum}
\item The emptiness checking can be achieved in linear time, and 
the membership checking can also be achieved in linear time if the largest
delay of the delay blocks is constant, and the largest timestamp in the word 
is of the order of the length of the word plus the automaton size.
\item The universality, containment in other untimed ADB languages, and
equivalence to regular languages are undecidable.
\end{compactenum}
\end{theorem}
\begin{proof}
The first item follows from Proposition~\ref{proposition:Emptiness} and Proposition~\ref{prop:Membership}.
The second item follows from Proposition~\ref{proposition:Universality}, 
Corollaries~\ref{cor:EquivReg} and~\ref{cor:Contain}.
\qed
\end{proof}

\subsection{Model Checking}
In this section we study the containment of a given untimed language of
an ADB within a given regular language.

\begin{proposition}[\tnote{Checking emptiness of  intersection with a regular language}]
\label{prop:IntRegEmpty}
Let $\A$ be an ADB, and let $\A_r$ be a regular finite state automaton. 
It is decidable to check emptiness of $\ulan(\A)\cap \ulan(\A_r)$ in time
$O(n_{\A}\cdot n_{\A_r}^{2M+1} + 
m_{\A} \cdot m_{\A_r}^{M+1} +  m_{\A}\cdot n_{\A_r})$ where
$n_{\A}, n_{\A_r} $ are the numbers of locations in $\A, \A_r$ and
$m_{\A}, m_{\A_r} $  the numbers of edges in $\A, \A_r$ respectively, and 
 $M$ is the largest delay of a delay block in $\A$.
\end{proposition}
\begin{proof}
We construct a non-deterministic  ADB  $\A^{\dagger}$ such that
$\ulan(\A^{\dagger}) = \ulan(\A)\cap \ulan(\A_r)$
and then apply Proposition~\ref{proposition:Emptiness}.

Let $M$ be the largest delay of a delay block in $\A$.
The automaton $\A^{\dagger}$ will simulate the executions of $\A$; and of $\A_r$ 
simultaneously for
the current time, and for time upto $M$ time units in the future.
That is, the automaton $\A^{\dagger}$ is able to verify that some  execution of $\A_r$
is such that
 $\A_r$ first generates the output symbols corresponding to the current time outputs in $\A$,
then generates symbols corresponding to current time plus one in $\A$, and so on.
To concurrently simulate executions of $\A_r$ corresponding to $M+1$ time
points, the automaton will maintain a tuple of locations.
The tuple  will have $2M+2$ components:
\begin{compactenum}
\item
The first component will correspond to a location of $\A$, and is used  to simulate executions
of $\A$.
\item 
The next $M+1$ components will correspond to locations of $\A_r$ used for
concurrently simulating $\A_r$ corresponding to the 
current time point, and the next $M$
time points.
\item 
The final $M$ components correspond to the ``guesses'' on the 
locations of $\A_r$ for the final locations
of the initial and following $M-1$ timepoints.
Whenever time elapses in $\A$ via an explicit $\tick$ transition, we verify that our ``guesses'' for 
the ending locations
of $\A_r$ are correct.
\end{compactenum}
%
%\vspace*{-6mm} 
%\begin{widetext}
%\vspace*{4mm}
Formally,
let $\A = (L^{\A}, D^{\A}, \Sigma, \delta^{\A}, l_s^{\A}, L_f^{\A})$ and
  $\A_r = (L^{\A_r}, \Sigma, \delta^{\A_r}, l_s^{\A_r}, L_f^{\A_r})$. 
The ADB  $\A^{\dagger}$ is as follows:
\begin{itemize}
\item
  The location set  $L^{\A^{\dagger}}$ is 
  $L^{\A}  \times \left(L^{\A_r}\right)^{2M+1}\, \cup\, \set{l_s^{\A^{\dagger}}}$ where
  $M$ is the largest delay of a delay block in $\A$, and $l_s^{\A^{\dagger}}$ is a new
  location.
%and  $l_{\tick}^{\A}$ is a new location component (\emph{i.e.} it is not present in
%  $\A$).
\item
  The initial location is $l_s^{\A^{\dagger}}$.
\item
  The transition function $\delta^{\A^{\dagger}}$ is as follows:
  \begin{itemize}
    \item
      $\delta^{\A^{\dagger}}(l_s^{\A^{\dagger}}, \epsilon) = 
      \left\{
        \tuple{l_s^{\A}, l_s^{\A_r}, l_1^{\A_r},\dots, l_M^{\A_r}, l_1^{\A_r},
          \dots, l_M^{\A_r}} \ \left|\ 
        \begin{array}{l}
        l_s^{\A} \text{ is the initial location of }
        \A,\\
        \l_s^{\A_r} \text{ is the initial location of }
        \A_r \\
        \text{and }l_j^{\A_r} \in L \text{ for } M\geq j \geq 1
      \end{array}
      \right.
      \right\}$\\
      The locations $l_s^{\A_r}, l_1^{\A_r},\dots, l_M^{\A_r}$ for $\A_r$ correspond
      the starting location for the current time, and for the following
      $M$ time points.
      The guessed  locations $l_1^{\A_r},\dots, l_M^{\A_r}$ are explicitly stored
      as the last $M$ components (to be verified later).
      There is an $\epsilon$-transition from  $l_s^{\A^{\dagger}}$ to various locations
      corresponding to all the possible guesses for the starting locations for
      times $1$ through $M$.
  
    \item
        $\delta^{\A^{\dagger}}
      (\tuple{l_0^{\A}, l_0^{\A_r}, l_1^{\A_r},\dots, l_M^{\A_r}, 
        \tilde{l}_1^{\A_r},
        \dots, \tilde{l}_M^{\A_r}}, \sigma, \block{t}) =$\\
      \[
      \left\{
        \tuple{{l_0'}^{\A}, {l_0'}^{\A_r}, {l_1'}^{\A_r},
          \dots, {l_M'}^{\A_r},
          \tilde{l}_1^{\A_r},
          \dots, \tilde{l}_M^{\A_r}}\ \left|\
          \begin{array}{l} 
            {l_0'}^{\A} \in 
            \delta^{\A}(l_0^{\A}, \tuple{\sigma, \block{t}}),\\
            {l_t'}^{\A_r} \in  \delta^{\A_r}(l_t^{\A_r}, \sigma), 
            \text{ and }\\
            {l_j'}^{\A_r} = l_j^{\A_r} \text{ for } j\neq t
          \end{array}
        \right.
      \right\}
      \]
      This transition corresponds to the case when $\A$ transitions on
     $\tuple{\sigma,\block{t}}$ from location $l_0^{\A}$.
      In the location $\tuple{l_0^{\A}, l_0^{\A_r}, l_1^{\A_r},\dots, l_M^{\A_r}, 
        \tilde{l}_1^{\A_r},
        \dots, \tilde{l}_M^{\A_r}}$ of $\A_r'$,  the component corresponding
      to $t$ time units in the future is updated.
      The location component of $\A$ is also updated.
      The rest of the components remain the same.

    \item
      $\delta^{\A^{\dagger}}
      (\tuple{l_0^{\A}, l_0^{\A_r}, l_1^{\A_r},\dots, l_M^{\A_r}, 
        \tilde{l}_1^{\A_r},
        \dots, \tilde{l}_M^{\A_r}}, \epsilon) =$\\
      \[
        \left\{
          \tuple{{l_0'}^{\A}, {l_0'}^{\A_r}, {l_1'}^{\A_r},
            \dots, {l_M'}^{\A_r},
            \tilde{l}_1^{\A_r},
            \dots, \tilde{l}_M^{\A_r}}\ \left|\
            \begin{array}{l} 
              {l_0'}^{\A} \in 
              \delta^{\A}(l_0^{\A}, \epsilon)\,\cup\,\set{l_0^{\A}}
            \text{ and }\\
            {l_j'}^{\A_r} \in  \delta^{\A_r}(l_j^{\A_r}, \epsilon)\, \cup\, 
            \set{l_j^{\A_r}}
              \text{ for }   M\geq j \geq 0
            \end{array}
          \right.
        \right\} 
     \]
     The $\epsilon$-transitions of $\A^{\dagger}$ correspond to the case when either
    $\A$ or $\A_r$ make $\epsilon$-transitions.
    The automaton $\A_r$ can make $\epsilon$-transitions either at the current time,
    or be supposed to make them in the future.

  \item 
          $\delta^{\A^{\dagger}}
      (\tuple{l_0^{\A}, l_0^{\A_r}, l_1^{\A_r},\dots, l_M^{\A_r}, 
        \tilde{l}_1^{\A_r},
        \dots, \tilde{l}_M^{\A_r}}, \tick) =$\\
      \begin{compactitem}
      \item  If  $ l_0^{\A_r}= \tilde{l}_{1}^{\A_r}$ then
        \[
        \left\{
          \tuple{{l_0'}^{\A}, {l_0'}^{\A_r}, {l_1'}^{\A_r},
            \dots, {l_M'}^{\A_r},
            \tilde{l'}_1^{\A_r},
            \dots, \tilde{l'}_M^{\A_r}}\ \left|\
            \begin{array}{l} 
              {l_0'}^{\A} \in 
              \delta^{\A}(l_0^{\A}, \tick), \\
              {l_j'}^{\A_r} = l_{j+1}^{\A_r} 
              \text{ and } \tilde{l'}_j^{\A_r} = \tilde{l}_{j+1}^{\A_r} 
              \text{ for } 0\leq j\leq M-1,\\
              {l_M'}^{\A_r}\in L^{\A_r} \text{ and } 
              \tilde{l'}_M^{\A_r}  \in L^{\A_r} \text{ with } 
              {l_M'}^{\A_r} =  \tilde{l'}_M^{\A_r}  \\
            \end{array}
          \right.
        \right\}
    \]\\
    $\qquad$

\item 
        $\emptyset$ otherwise.
      \end{compactitem}

    When time advances from time $t$ to $t+1$ in $\A^{\dagger}$ (and in $\A$), 
    we need to verify that the initial guess $\tilde{l}_{1}^{\A_r}$
    was correct
    (recall that $\tilde{l}_{1}^{\A_r}$ was guessed previously in time $t$ to be the location
    in $\A_r$ at the beginning of time $t+1$.
    Also, since one time unit has passed, the guess  $\tilde{l}_{j+1}^{\A_r}$ for the
    earlier $j+1$-th future time unit now becomes the guess corresponding to
    $j$-th future time unit.
    Similarly, the location component $l_{j+1}^{\A_r}$ corresponding to the 
    $j+1$-th future time unit now becomes the location component corresponding to
    $j$-th future time unit.
    A new guess for the starting location for the $M$-th future time unit is also
    chosen.
              
    \end{itemize}

  \item
    The set of final locations is 
    \[
    \left\{
    \tuple{l_0^{\A}, l_0^{\A_r}, l_1^{\A_r},\dots, l_M^{\A_r}, 
      \tilde{l}_1^{\A_r},
      \dots, \tilde{l}_M^{\A_r}} \left|\ 
      \begin{array}{l} 
        l_0^{\A} \in L_f^{\A} \text{ and } 
        l_M^{\A_r} \in L_f^{\A_r} \text{ and } \\
        l_j^{\A_r} = \tilde{l}_{j+1}^{\A_r}  \text{ for } 0\leq j \leq M-1
      \end{array}
    \right.
  \right\}
  \]
  $\A^{\dagger}$ ensures that the automaton $\A_r$ ends up in an accepting location 
  at the end.
  Also, $\A^{\dagger}$ checks that the guesses for the starting locations at each of the
  $M$ future timepoints were correct.
\end{itemize}
%\end{widetext}    
The automaton $\A^{\dagger}$ has an accepting run $\pi^{\dagger}$ only if both of the conditions
hold.
\begin{compactenum}
\item
  The run $\pi^{\dagger}$ corresponds to a matching generating run $\pi$ in $\A$.

  % The word $w^u$ is obtained by traversing the automaton $\A$ going from
  % the initial location to a final location (equivalently, the untimed word
  % obtained from a corresponding path in an automaton $\widehat{\A}$
  % obtained from replacing each delay block in $\A$ by a delay block of
  % delay 0).
  % Let us call the set of paths which correspond to $w^u$ as  $p(w^u)$.
  % For example, the word $w^u$ for the automaton $\A_1$ 
  % of Figure~\ref{figure:example-one} would belong to the language $(abc)^*$, and
  % the path set $p((abc)^n)$ would be $\set{(l_0 l_1 l_2)^n}$.
\item
  The untimed word $\untime(w_{\pi})$
  can be generated from $\A_r$, where
   $w_{\pi}$ denotes the timed word output by $\A$ corresponding to the generating path $\pi$.
   That is, the automaton  $\A_r$ generates the output symbols in $w_{\pi}$ in 
   the \emph{right order} corresponding to $\untime(w_{\pi})$.

% {right order} as generated
%   by some $\A$-path in $p(w^u)$.
%   For example, for $w^u = (abc)^n$ in $\A_1$ of Figure~\ref{figure:example-one},
%   the automaton  $\A_r$ would have to generate $a^n b^n c^n$.

\end{compactenum} 

Thus, we have that $\ulan(\A^{\dagger}) = \ulan(\A)\cap \ulan(\A_r)$.,
and that $\ulan(\A)\cap \ulan(\A_r)$ is non-empty iff
$\ulan(\A^{\dagger}) $ is non-empty.
The number of locations in $\A^{\dagger}$ is $1+ n_{\A}\cdot n_{\A_r}^{2M+1}$.
%, where
%$n_{\A}$ is the number of locations in $\A$ and $ n_{\A_r}$ the number of locations in
%$\A_r$.
The number of edges is 
$n_{\A_r}^M + m_{\A} \cdot m_{\A_r}^{M+1} +  m_{\A}\cdot n_{\A_r}$.
% where
%$ m_{\A_r}$ and $m_{\A}$ are the number of edges in $\A_r$ and $\A$ respectively.
Thus, emptiness of $\ulan(\A)\cap \ulan(\A_r)$ can be checked in time
$O( n_{\A}\cdot n_{\A_r}^{2M+1} + 
m_{\A} \cdot m_{\A_r}^{M+1} +  m_{\A}\cdot n_{\A_r})$.
\qed
\end{proof}

%%\begin{corollary}[Checking containment in regular language]
%%\label{cor:ContainReg}
%%Let $\A$ be an ADB and let $\R$ be a regular language.
%%It is decidable to check whether $\ulan(\A) \subseteq \R$.
%%\end{corollary}

\begin{theorem}[\tnote{Model Checking}]
Let $\A$ be an ADB and $\A_r$ be a regular non-deterministic 
finite-state automaton, 
and let $n_{\A}, n_{\A_r} $ be the numbers of locations in $\A, \A_r$ and
$m_{\A}, m_{\A_r} $  be the numbers of edges in $\A, \A_r$ respectively, and 
$M$ be the largest delay of a delay block in $\A$.
The following assertions hold:
\begin{compactenum}
\item The problem whether $\ulan(\A)\subseteq \ulan(\A_r)$ can be checked in time
$O(n_{\A}\cdot 2^{n_{\A_r} \cdot (2M+1)} + 
m_{\A} \cdot 2^{m_{\A_r} \cdot (M+1)} +  m_{\A}\cdot 2^{n_{\A_r}})$.

\item If $M$ is constant, then the problem of whether $\ulan(\A)\subseteq \ulan(\A_r)$ 
is PSPACE-complete.
\end{compactenum}
\end{theorem}
\begin{proof}
The result is obtained as follows: (i)~we first obtain an automata for the complement 
of the language of $\A_r$ (and the automata can be obtained by first determinizing 
$\A_r$ and then complementing, and thus has at most $2^{n_{\A_r}}$ locations and 
$2^{m_{\A_r}}$ edges); and 
(ii)~then check emptiness of intersection using Proposition~\ref{prop:IntRegEmpty}.
This gives the result for the first item.
For the second item we note that if $M$ is constant, then we obtain an exponential 
size automata and the emptiness check is achieved by checking reachability to 
a final location (in non-deterministic log-space over an exponential graph).
This gives the PSPACE upper bound, and the PSPACE lower bound follows from 
the fact the containment is PSPACE hard for regular non-deterministic finite-state 
automata (which are special cases of ADBs).
The desired result follows.
\qed
\end{proof}

\renewcommand*{\bibfont}{\raggedright}
\printbibliography

%\clearpage
\section{Appendix}

\smallskip\noindent
\textbf{Proof of Proposition~\ref{prop:NonClosureInt}.}
\begin{proof}
Consider the language
  \[
  \lan^{\dagger} = \set{\kappa_0(w \#) \kappa_1(w \#)
    \kappa_2(w \#)  \dots \kappa_n(w \#) \mid w\in\set{a,b}^* \text{ and } n\geq 0}
  \]
  We show $\lan^{\dagger}$ is not an ADB language, and that
  there exist ADBs $\A_1$ and $\A_2$ such that
  $\lan^{\dagger} = \lan(\A_1) \cap   \lan(A_2)$.
 \begin{enumerate}
  \item
    ($\lan^{\dagger}$ is not an ADB language.)\\
Suppose the claim is false, \emph{ie.} let $\lan^{\dagger} = \lan(\A)$.
Let the number of locations in $\A$ be $K$.
 Consider a timed word $w_{\dagger} =
  \kappa_0(w \#) \kappa_1(w \#)
  \kappa_2(w \#)  \dots \kappa_{K+2}(w \#) $ with
  $|w| > K$.
  Let $r_{\dagger}$ be the generating run for $w_{\dagger}$
  Using the pumping lemma for ADB  runs (Proposition~\ref{prop:PumpingLemmaRuns}),  
we can show there exists  a subrun $r_p$ of
  $r_{\dagger}$ such that 
  (1)~the subrun $r_p$ contains at least one output symbol
  transition, and
  (2)~the subrun contains at most $K$ output symbol transitions; and
  (3)~for $r_{\dagger} = r_0\circ r_p \circ r_1$, we have that
  $r_0 \circ r_1$ is also a generating run for $\A^{\dagger}$ (\emph{i.e.},
  we pump down $r_p$).
  Let $w_{01}$ be the output word corresponding to
  the generating run  $r_0 \circ r_1$.
  Because of the constraints on $r_p$, we have that
  $w_{01}$ contains at least one, and at most $K$ output symbols
  less than $w$.
  It can be checked that this means that $w_{01}$ is not a member of $\lan^{\dagger} $,
  a contradiction. 
Thus, $\lan^{\dagger}$ is not an ADB language.
% Fix some string $\overline{u} \in \set{a,b}^* $ with $|u| \geq M$.
% Consider the following word in $\lan^{\dagger} $
% \[
% w^{\dagger} = \kappa_0(u \#) \kappa_1(u \#)
%     \kappa_2(u \#)  \dots \kappa_{2M+2}(u \#) 
% \]
% Let  $\overline{w}_0^{\dagger},\dots,\overline{w}_m^{\dagger}$ be non-overlapping
% substrings  of $w^{\dagger} $
% with $m\geq 0$ and $\sum_{j=0}^m |\overline{w}_j^{\dagger}| \, < M  $ as in the pumping lemma
% (Proposition~\ref{prop:PumpingLemma}).
% We have that there must exist $\alpha$ with $0\leq \alpha \leq 2M +2$ such that
% for all $0\leq j \leq m$ the substrings 
% $\kappa_{\alpha}(u \#)$ and $ \overline{w}_j^{\dagger}$ have no symbols in common, \emph{ie.}
% none of the timestamps of the elements from $\overline{w}_0^{\dagger},\dots,\overline{w}_m^{\dagger}$
% equals $\alpha$. 
% This claim holds because for each  substring $\overline{w}_j^{\dagger}$, the set of timestamps of the elements 
%  of the string $\overline{w}_j^{\dagger}$ is a subset of $\set{k, k+1}$ for some $k$ (note that
%   $|\overline{w}_j^{\dagger}| < M$ and $|\kappa_{i}(u \#)| > M$ for every $i$, thus
%  each $\overline{w}_j^{\dagger}$ can intersect  at most two consecutive
%  $\kappa_{i}(u \#)$).

% Pumping the string according to the pumping lemma, we get a string where 
% no element of $\kappa_{\alpha}(u \#)$ gets pumped but there exists a 
% $\kappa_{\beta}(u \#)$ such that some element of it gets pumped.
% This pumped string clearly does not belong in $\lan^{\dagger}$, contradicting the
% pumping lemma.

\item
 ($\lan^{\dagger}$ is the intersection of two ADB languages).\\
Consider the following two  languages $  \lan^{\dagger}_1  $ and
$  \lan^{\dagger}_2  $:
%\begin{align*}
 % &
%\begin{figure*}[!t]
%\begin{widetext}
\begin{equation}
\label{equation:lanone}
  \begin{array}{ll}
    \lan^{\dagger}_1   = & 
    \left\{\left.
        \begin{array}{l}
          \kappa_0(\overline{w}_0 \#)\mysep \kappa_1(\overline{w}_0 \#)\mysep
          \kappa_2(\overline{w}_2 \#)  \mysep \kappa_3(\overline{w}_2 \#)  \dots\\ 
          \kappa_{2n}(\overline{w}_{2n} \#) \mysep
          \kappa_{2n+1}(\overline{w}_{2n} \#) 
        \end{array}
      \right\vert \overline{w}_j \in\set{a,b}^*  \text{ for all } j \text{ and } n\geq 0
    \right\}\\
     & \mspace{180mu}\cup \\
      & \left\{\left.
        \begin{array}{l}
          \kappa_0(\overline{w}_0 \#)\mysep \kappa_1(\overline{w}_0 \#)\mysep
          \kappa_2(\overline{w}_2 \#)  \mysep \kappa_3(\overline{w}_2 \#)  \dots \\
          \kappa_{2n-2}(\overline{w}_{2n-2} \#) \mysep
          \kappa_{2n-1}(\overline{w}_{2n-1} \#) \mysep
          \kappa_{2n}(\overline{w}_{2n} \#) 
        \end{array}
      \right\vert  \overline{w}_j \in\set{a,b}^*  \text{ for all } j \text{ and } n\geq 0
    \right\}
  \end{array}\\%[10pt]
\end{equation}
and 
\begin{equation}
  \begin{array}{ll}
 \lan^{\dagger}_2  = &
    \left\{\left.
        \begin{array}{l}
          \kappa_0(\overline{w}_{-1} \#)\mysep \kappa_1(\overline{w}_1 \#)\mysep\kappa_2(\overline{w}_1 \#)\\
          \kappa_3(\overline{w}_3 \#)\mysep\kappa_4(\overline{w}_3 \#)\dots
          \kappa_{2n-1}(\overline{w}_{2n-1} \#) \mysep
          \kappa_{2n}(\overline{w}_{2n-1} \#) 
        \end{array}
      \right\vert  \overline{w}_j \in\set{a,b}^*  \text{ for all } j \text{ and } n\geq 0
    \right\}\\
   & \mspace{180mu}\cup\\
  & \left\{\left.
      \begin{array}{l}
        \kappa_0(\overline{w}_{-1} \#)\mysep \kappa_1(\overline{w}_1 \#)\mysep\kappa_2(\overline{w}_1 \#)\\
        \kappa_3(\overline{w}_3 \#)\mysep\kappa_4(\overline{w}_3 \#)\dots\\
        \kappa_{2n-1}(\overline{w}_{2n-1} \#) \mysep
        \kappa_{2n}(\overline{w}_{2n-1} \#) \mysep
        \kappa_{2n}(\overline{w}_{2n} \#)
      \end{array}
    \right\vert  \overline{w}_j \in\set{a,b}^*  \text{ for all } j \text{ and } n\geq 0
  \right\}
 \end{array}
%
%
%\end{align*}
\end{equation}
%\end{figure*}
%\end{widetext}
%
The language $\lan^{\dagger}_1 $ contains all (discrete) times words
of the form $\kappa_0(\overline{w}_0 \#) \kappa_1(\overline{w}_1 \#)
    \kappa_2(\overline{w}_2 \#)  \dots $ such that $\overline{w}_{2j} = \overline{w}_{2j+1}$ for all $j$.
The language $\lan^{\dagger}_2 $ contains all (discrete) times words
of the form $\kappa_0(\overline{w}_0 \#) \kappa_1(\overline{w}_1 \#)
    \kappa_2(\overline{w}_2 \#)  \dots $ such that $\overline{w}_{2j+1} = \overline{w}_{2j+2}$ for all $j$.

    We show both  $\lan^{\dagger}_1 $ and  $\lan^{\dagger}_2 $  to be  ADB languages.
    Consider the ADB $\A_1^{\dagger}$ in Figure~\ref{figure:complementone} which has $l_0$ as the
    initial location and $l_5$ and $l_7$ as the accepting locations.
    It can be seen that $\lan(\A_1^{\dagger}) =   \lan^{\dagger}_1$.
    \begin{figure*}[t]
      \strut\centerline{\input Figures/delaycomplement1.eepic}
      \caption{Automaton $\A_1^{\dagger}$ with delay blocks such that $\lan(\A_1^{\dagger}) =   \lan^{\dagger}_1$.}
      \label{figure:complementone}
    \end{figure*}
Consider the right hand side of the automaton from the starting location $l_0$.
This portion of the ADB results in the output part $\set{\kappa_0(\ol{w}) \mid \ol{w} \in \set{a,b}^*}$
\emph{i.e.} it is used to add the last ``unmatched'' segment in the second part after the union
in Equation~\ref{equation:lanone}.
The left hand side of the automaton is used to generate the segments
$ \kappa_{2i}(\overline{w}_{2i} \#)\mysep \kappa_{2i+1}(\overline{w}_{2i} \#)$.

The ADB  $\A_2^{\dagger}$ in Figure~\ref{figure:complementtwo} has $l_9$ as the (only) starting location,
and includes the locations and transitions of  $\A_1^{\dagger}$ as shown in the Figure.
The $\tick$ transition from $l_{10}$ goes to the location $l_0$.
The accepting locations are $l_5$, $l_7$ and $l_{10}$. 
The automaton just ``shifts'' the checks of $\A_1^{\dagger}$ by one time unit, by
the transitions between $l_9$ and $l_{10}$. 
It can be seen that $\lan(\A_2^{\dagger}) =   \lan^{\dagger}_2$.

 \begin{figure*}[t]
      \strut\centerline{\setlength{\unitlength}{0.00043745in}
\begingroup\makeatletter\ifx\SetFigFont\undefined%
\gdef\SetFigFont#1#2#3#4#5{%
  \reset@font\fontsize{#1}{#2pt}%
  \fontfamily{#3}\fontseries{#4}\fontshape{#5}%
  \selectfont}%
\fi\endgroup%
{\renewcommand{\dashlinestretch}{30}
\begin{picture}(9744,3009)(0,-10)
\path(912,2982)(1407,2982)(1407,2667)
	(912,2667)(912,2982)
\path(912,597)(1407,597)(1407,282)
	(912,282)(912,597)
\put(4084,1520){\ellipse{854}{584}}
\put(4084,1520){\ellipse{944}{674}}
\path(2577,1677)(3072,1677)(3072,1362)
	(2577,1362)(2577,1677)
\put(6405,1545){\ellipse{854}{584}}
\put(1199,1563){\ellipse{854}{584}}
\dashline{60.000}(9732,12)(9732,2982)(5952,2982)
	(5952,12)(9732,12)
\path(4557,1497)(5952,1497)
\blacken\path(5772.000,1437.000)(5952.000,1497.000)(5772.000,1557.000)(5826.000,1497.000)(5772.000,1437.000)
\path(3072,1497)(3612,1497)
\blacken\path(3432.000,1437.000)(3612.000,1497.000)(3432.000,1557.000)(3486.000,1497.000)(3432.000,1437.000)
\path(1407,2802)(1410,2800)(1417,2795)
	(1430,2787)(1448,2776)(1471,2761)
	(1499,2743)(1529,2723)(1560,2702)
	(1590,2680)(1619,2660)(1646,2640)
	(1671,2621)(1693,2603)(1713,2587)
	(1730,2571)(1745,2555)(1759,2540)
	(1771,2525)(1782,2509)(1792,2494)
	(1801,2478)(1809,2461)(1816,2444)
	(1822,2426)(1828,2407)(1832,2388)
	(1835,2368)(1838,2347)(1839,2326)
	(1840,2305)(1839,2285)(1838,2264)
	(1835,2243)(1832,2223)(1828,2203)
	(1823,2183)(1818,2164)(1811,2146)
	(1805,2127)(1797,2108)(1788,2090)
	(1779,2070)(1768,2050)(1756,2029)
	(1742,2006)(1727,1982)(1710,1956)
	(1691,1928)(1670,1899)(1649,1868)
	(1627,1838)(1606,1809)(1587,1783)
	(1571,1760)(1542,1722)
\blacken\path(1603.504,1901.492)(1542.000,1722.000)(1698.898,1828.691)(1618.441,1822.164)(1603.504,1901.492)
\path(1407,462)(1411,464)(1418,469)
	(1431,478)(1450,491)(1473,507)
	(1499,526)(1526,546)(1554,567)
	(1579,587)(1603,606)(1624,624)
	(1643,642)(1660,659)(1674,675)
	(1686,691)(1697,708)(1707,724)
	(1716,742)(1723,760)(1730,780)
	(1736,800)(1741,821)(1745,843)
	(1749,865)(1751,888)(1751,911)
	(1751,934)(1750,957)(1748,979)
	(1744,1000)(1740,1020)(1735,1040)
	(1729,1058)(1722,1076)(1715,1092)
	(1706,1108)(1696,1123)(1685,1137)
	(1673,1152)(1659,1167)(1642,1182)
	(1624,1198)(1603,1214)(1580,1231)
	(1556,1248)(1532,1264)(1510,1280)
	(1489,1293)(1452,1317)
\blacken\path(1635.664,1269.383)(1452.000,1317.000)(1570.362,1168.708)(1557.709,1248.432)(1635.664,1269.383)
\path(777,1632)(775,1634)(771,1640)
	(765,1649)(754,1663)(740,1683)
	(723,1707)(703,1734)(681,1765)
	(658,1798)(635,1831)(612,1863)
	(590,1895)(569,1926)(550,1955)
	(532,1982)(516,2007)(502,2031)
	(489,2053)(477,2074)(466,2094)
	(456,2113)(448,2131)(440,2149)
	(431,2169)(424,2189)(417,2208)
	(411,2228)(406,2248)(402,2268)
	(398,2288)(396,2308)(394,2328)
	(393,2348)(394,2368)(395,2388)
	(398,2407)(401,2426)(405,2444)
	(411,2462)(417,2479)(424,2495)
	(432,2511)(441,2526)(451,2540)
	(462,2554)(474,2568)(487,2582)
	(501,2596)(517,2610)(535,2624)
	(555,2639)(578,2655)(603,2671)
	(630,2688)(660,2707)(692,2725)
	(725,2745)(759,2764)(792,2782)
	(823,2799)(851,2814)(874,2827)
	(891,2836)(902,2842)(909,2845)(912,2847)
\path(777,1542)(775,1540)(771,1535)
	(765,1526)(754,1513)(740,1495)
	(723,1473)(703,1448)(682,1420)
	(659,1390)(636,1359)(613,1329)
	(592,1299)(571,1271)(552,1244)
	(535,1219)(519,1196)(505,1174)
	(493,1153)(481,1134)(471,1115)
	(462,1097)(454,1079)(447,1062)
	(439,1041)(432,1021)(427,1000)
	(422,979)(418,958)(415,937)
	(413,915)(413,894)(413,872)
	(415,851)(418,830)(422,809)
	(428,789)(434,770)(441,752)
	(450,735)(459,718)(469,702)
	(480,687)(492,672)(504,659)
	(516,646)(530,634)(546,621)
	(563,608)(582,595)(603,581)
	(626,567)(652,552)(680,536)
	(709,520)(740,504)(771,487)
	(802,472)(830,457)(856,445)
	(877,434)(893,426)(903,421)
	(909,418)(912,417)
\path(1587,1497)(2577,1497)
\path(12,1587)(777,1587)
\blacken\path(597.000,1527.000)(777.000,1587.000)(597.000,1647.000)(651.000,1587.000)(597.000,1527.000)
\path(6807,1542)(6810,1542)(6816,1543)
	(6828,1545)(6846,1548)(6870,1552)
	(6900,1556)(6934,1562)(6972,1568)
	(7012,1575)(7053,1582)(7093,1589)
	(7133,1596)(7170,1603)(7205,1610)
	(7238,1618)(7269,1624)(7297,1631)
	(7324,1638)(7349,1646)(7372,1653)
	(7395,1661)(7416,1669)(7437,1677)
	(7459,1687)(7481,1697)(7503,1708)
	(7525,1720)(7547,1733)(7570,1748)
	(7595,1764)(7620,1781)(7646,1800)
	(7674,1821)(7703,1843)(7732,1866)
	(7761,1889)(7789,1911)(7815,1932)
	(7838,1951)(7856,1966)(7887,1992)
\blacken\path(7787.642,1830.358)(7887.000,1992.000)(7710.529,1922.301)(7790.460,1911.031)(7787.642,1830.358)
\path(6807,1452)(6810,1451)(6817,1449)
	(6829,1446)(6847,1441)(6871,1434)
	(6900,1426)(6933,1417)(6969,1406)
	(7006,1395)(7042,1384)(7078,1373)
	(7112,1362)(7144,1351)(7175,1341)
	(7203,1330)(7229,1320)(7255,1309)
	(7279,1299)(7302,1288)(7324,1276)
	(7347,1264)(7368,1253)(7388,1241)
	(7410,1228)(7431,1214)(7454,1200)
	(7478,1183)(7503,1166)(7529,1147)
	(7557,1126)(7587,1103)(7619,1079)
	(7652,1054)(7685,1028)(7719,1001)
	(7752,975)(7783,950)(7811,928)
	(7836,908)(7855,893)(7887,867)
\blacken\path(7709.464,933.940)(7887.000,867.000)(7785.135,1027.074)(7789.210,946.455)(7709.464,933.940)
\path(7212,417)(7209,417)(7202,418)
	(7190,420)(7171,423)(7146,427)
	(7115,432)(7079,437)(7040,444)
	(6999,451)(6956,458)(6915,465)
	(6875,473)(6837,480)(6802,488)
	(6769,495)(6739,503)(6712,510)
	(6687,518)(6664,525)(6644,533)
	(6624,542)(6606,550)(6590,559)
	(6571,571)(6553,584)(6536,597)
	(6520,611)(6504,627)(6490,643)
	(6476,660)(6463,678)(6451,697)
	(6440,716)(6430,735)(6420,755)
	(6412,774)(6404,794)(6397,813)
	(6391,832)(6386,851)(6381,869)
	(6376,887)(6372,904)(6368,924)
	(6363,944)(6359,964)(6355,985)
	(6351,1007)(6347,1031)(6343,1057)
	(6339,1084)(6334,1114)(6330,1144)
	(6325,1174)(6322,1202)(6318,1227)(6312,1272)
\blacken\path(6395.263,1101.509)(6312.000,1272.000)(6276.316,1085.649)(6328.653,1147.105)(6395.263,1101.509)
\put(1092,2757){\makebox(0,0)[lb]{\smash{{\SetFigFont{8}{9.6}{\familydefault}{\mddefault}{\updefault}0}}}}
\put(1092,372){\makebox(0,0)[lb]{\smash{{\SetFigFont{8}{9.6}{\familydefault}{\mddefault}{\updefault}0}}}}
\put(2757,1452){\makebox(0,0)[lb]{\smash{{\SetFigFont{8}{9.6}{\familydefault}{\mddefault}{\updefault}0}}}}
\put(6267,1452){\makebox(0,0)[lb]{\smash{{\SetFigFont{8}{9.6}{\familydefault}{\mddefault}{\updefault}$l_0$}}}}
\put(1902,1272){\makebox(0,0)[lb]{\smash{{\SetFigFont{8}{9.6}{\familydefault}{\mddefault}{\updefault}$\#$}}}}
\put(4692,1272){\makebox(0,0)[lb]{\smash{{\SetFigFont{8}{9.6}{\familydefault}{\mddefault}{\updefault}$\tick$}}}}
\put(1137,1497){\makebox(0,0)[lb]{\smash{{\SetFigFont{8}{9.6}{\familydefault}{\mddefault}{\updefault}$l_9$}}}}
\put(3972,1452){\makebox(0,0)[lb]{\smash{{\SetFigFont{8}{9.6}{\familydefault}{\mddefault}{\updefault}$l_{10}$}}}}
\put(6897,1677){\makebox(0,0)[lb]{\smash{{\SetFigFont{8}{9.6}{\familydefault}{\mddefault}{\updefault}$\epsilon$}}}}
\put(6897,1137){\makebox(0,0)[lb]{\smash{{\SetFigFont{8}{9.6}{\familydefault}{\mddefault}{\updefault}$\epsilon$}}}}
\put(6312,327){\makebox(0,0)[lb]{\smash{{\SetFigFont{8}{9.6}{\familydefault}{\mddefault}{\updefault}$\tick$}}}}
\put(6042,2667){\makebox(0,0)[lb]{\smash{{\SetFigFont{10}{12.0}{\familydefault}{\mddefault}{\updefault}$\A_1^{\dagger}$}}}}
\put(507,2217){\makebox(0,0)[lb]{\smash{{\SetFigFont{8}{9.6}{\familydefault}{\mddefault}{\updefault}$a$}}}}
\put(552,822){\makebox(0,0)[lb]{\smash{{\SetFigFont{8}{9.6}{\familydefault}{\mddefault}{\updefault}$b$}}}}
\end{picture}
}}
      \caption{Automaton $\A_2^{\dagger}$ with delay blocks such that $\lan(\A_2^{\dagger}) =   \lan^{\dagger}_2$.}
      \label{figure:complementtwo}
    \end{figure*}
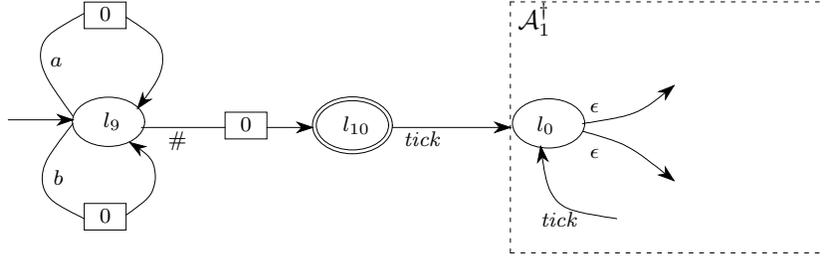
  \end{enumerate}
The proof for untimed languages is similar to that for timed languages (we use
the untimed language $\untime(\lan^{\dagger})$).
\qed
\end{proof}

\smallskip\noindent
\textbf{Proof of Proposition~\ref{prop:NonClosureComp}.}
\begin{proof}
Since $\lan(\A_1 \cap \A_2) = 
\overline{\overline{\lan(\A_1)} \cup \overline{\lan(\A_2)}}$, and since
ADBs are not closed under intersection (Proposition~\ref{prop:NonClosureInt}) but are closed under
union (Proposition~\ref{prop:UnionClosure}) we have
that ADBs are not closed under complementation (a similar proof applies to untimed languages).
\qed
\end{proof}

\smallskip\noindent
\textbf{Proof of Proposition~\ref{prop:Membership}.}
\begin{proof}
Let $w = \tuple{w_0^{\sigma}, w_0^t} \tuple{w_1^{\sigma}, w_1^t} \dots
\tuple{w_n^{\sigma}, w_n^t}$.
Let the end timestamp of $w$ be $\tend$ (\emph{i.e.}  $   w_n^t = \tend$).
We first construct a finite state deterministic  automaton $\A_w$ with just one path 
(corresponding to $w$) over the
alphabet 
$\Sigma_w = \set{\tick, \tuple{w_0^{\sigma}, w_0^t},  \tuple{w_1^{\sigma}, w_1^t}
\dots \tuple{w_n^{\sigma}, w_n^t}}$ as follows.
\begin{itemize}
\item 
  The set of locations is
  $\set{l_0, l_1, \dots, l_{\tend}} \, \cup \, 
  \set{l_{ \tuple{w_0^{\sigma}, w_0^t}}, l_{\tuple{w_1^{\sigma}, w_1^t}}
\dots l_{\tuple{w_n^{\sigma}, w_n^t}}}$.
\item 
  Let $ w = \ol{w}_0 \circ  \ol{w}_1 \circ \dots  \ol{w}_{\tend}$ where
 $\ol{w}_k$ is the substring of $w$ containing all the timestamped tuples
of time $k$ (let $\ol{w}_k$ be $\epsilon$ if $w$ does not contain any $k$-timestamped
tuples).
Let $i_k = \sum_{j=0}^k |\ol{w}_{j}|$ for $k = 0..\tend$.
Note that if $w$ contains a timestamp $k$, then
$ \tuple{w_{i_k}^{\sigma}, w_{i_k}^t}$ is the last tuple timestamped with
$k$.
Thus,
%\begin{widetext}
\[
w = 
\underset{\text{timestamp } 0}{\underbrace{\tuple{w_0^{\sigma}, w_0^t} \tuple{w_1^{\sigma}, w_1^t} \dots 
\tuple{w_{i_0}^{\sigma}, w_{i_0}^t}}}
\underset{\text{timestamp } 1}{\underbrace{
\tuple{w_{i_0+1}^{\sigma}, w_{i_0+1}^t}\dots
\tuple{w_{i_1}^{\sigma}, w_{i_1}^t}}} \dots
\underset{\text{timestamp } \tend}{\underbrace{
\tuple{w_{i_{\tend -1}+1}^{\sigma}, w_{i_{\tend-1}+1}^t}\dots
\tuple{w_{I_\tend}^{\sigma}, w_{i_{\tend}^t}}}}
\]
%\end{widetext}
Note that $w_{j}^t = k$  for all $i_{k-1}+1 \leq j \leq i_k$ in the above representation.
The finite state automaton $\A_w$ will generate the language
$ \ol{w}_0 \circ \tick \circ\,  \ol{w}_1 \circ\tick \circ \dots  \tick \circ\, \ol{w}_{\tend} \circ
\tick^*$.
\item 
The starting location  is $l_0$ and the only accepting location is
$l_{\tuple{w_{I_\tend}^{\sigma}, w^t_{i_{\tend}}}}$.
We refer to the accepting location $l_{\tuple{w_{I_\tend}^{\sigma}, w^t_{i_{\tend}}}}$ as
the location $l_{\bot}$.
\item
The set of transitions is as follows.
For all $k\geq -1$
\begin{itemize}
\item 
  If there is a timestamp $k+1$ in $w$, then
  there are the following transitions.
  \begin{itemize}
  \item 
    The transition
    $l_{k+1}\stackrel{\tuple{w_{i_k+1}^{\sigma}, w_{i_k+1}^t}}{\longrightarrow}
    l_{\tuple{w_{i_k+1}^{\sigma}, w_{i_k+1}^t}}$
    (letting $i_{-1} = -1$).
  \item 
    For  all $j$ such that $ i_k +1 < j \leq  i_{k+1}$, we have the transition
    $l_{j-1} \stackrel{\tuple{w_{j}^{\sigma}, w_{j}^t}}{\longrightarrow} i_{j}$.
  \item 
    If $k+1$ is not the greatest timestamp in $w$, then
    the transition
   $ l_{\tuple{w_{i_{k+1}}^{\sigma}, w_{i_{k+1}}^t}} \stackrel{\tick}{\longrightarrow} l_{k+2}$.
 \item 
   The looping  transition $l_{\tuple{w_{I_\tend}^{\sigma}, w^t_{i_{\tend}}}}
     \stackrel{\tick}{\longrightarrow}
l_{\tuple{w_{I_\tend}^{\sigma}, w^t_{i_{\tend}}}}$.
 \end{itemize}
\item 
   If there is no timestamp   $k+1$ in $w$, but
   there is a timestamp larger than $k+1$ in $w$, then the transition
   $l_{k+1}  \stackrel{\tick}{\longrightarrow} l_{k+2}$.

\end{itemize}
It can be confirmed that the
 finite state automaton $\A_w$  generates the  strings
$ \ol{w}_0 \circ \tick \circ\,  \ol{w}_1 \circ\tick \circ \dots  \tick \circ\, \ol{w}_{\tend}
\circ \tick^*$.

\end{itemize}

We construct a (non-deterministic)  ADB  $\A^{\ddagger}$ based on
$\A$ and $\A_w$ such that
$\A^{\ddagger}$ has an accepting  path  iff
$\A$ outputs the timed word $w$.
Let $M$ be the largest delay of a delay block in $\A$.
The automaton $\A^{\ddagger}$ will simulate the executions of $\A$; and of $\A_w$ 
simultaneously for
the current time, and for time upto $M$ time units in the future.
That is, the automaton $\A^{\ddagger}$ is able to verify that $\A_w$ 
 first generates the output symbols corresponding to the current time outputs in $\A$,
then generates symbols corresponding to current time plus one in $\A$, and so on.
To concurrently simulate executions of $\A_w$ corresponding to $M+1$ time
points, the automaton will maintain a tuple of locations.
The tuple  will have $M+2$ components:
\begin{enumerate}
\item
The first component will correspond to a location of $\A$, and is used  to simulate executions
of $\A$.
\item 
The next $M+1$ components will correspond to locations of $\A_w$ used for
concurrently simulating $\A_w$ corresponding to the 
current time point, and the next $M$
time points.
\end{enumerate}
Formally,
let $\A = (L^{\A}, D^{\A}, \Sigma, \delta^{\A}, l_s^{\A}, L_f^{\A})$. 
The ADB  $\A^{\ddagger}$ is as follows:
\begin{itemize}
\item
  The location set   $L^{\A^{\ddagger}}$ is 
  $\left(L^{\A} \cup \set{l_{\tick}^{\A}}\right)\times \left(L^{\A_w}\right)^{M+1}$ where
  $M$ is the largest delay of a delay block in $\A$, where
  $l_{\tick}^{\A}$ is a new location component (\emph{i.e.} it is not present in
  $\A$).
\item
  The initial location is 
  $\tuple{l_s^{\A},l_0^{\A_w},  l_1^{\A_w}, \dots, l_{M}^{\A_w}}$
  (we have added the superscript $\A_w$ to the locations of $\A_w$ for clarity).
\item 
  The set of accepting locations is 
  $\left(L_f^{\A}  \cup \set{l_{\tick}^{\A}}\right)\times \underset{M+1 \text{ occurences}}{\underbrace{
  \set{ l_{\bot}^{\A_w}} \times  \set{ l_{\bot}^{\A_w}} \times \dots
  \set{ l_{\bot}^{\A_w}}}}$.\\
(Recall that $ l_{\bot}^{\A_w}$ denotes the location $l_{\tuple{w_{I_\tend}^{\sigma}, w^t_{i_{\tend}}}}$,
and is the end sink accepting location of $\A_w$.)
\item The transition function $\delta^{\ddagger}()$ is as follows.
%\begin{widetext}
  \begin{itemize}
  \item $\delta^{\A^{\ddagger}}(
    \tuple{
      l^{\A},  l_{\tuple{0}}^{\A_w}, l_{\tuple{1}}^{\A_w}, \dots l_{\tuple{M}}^{\A_w}
    }, \epsilon) =
    \delta^{\A}(l^{\A}, \epsilon) \times \set{ l_{\tuple{0}}^{\A_w}}
    \times \set{ l_{\tuple{1}}^{\A_w}} \times \dots
      \set{ l_{\tuple{M}}^{\A_w}}
    $ (the locations $ l_{\tuple{j}}^{\A_w}$ are unrelated to the locations
    $ l_{j}^{\A_w}$).

  \item 
    $\delta^{\A^{\ddagger}}(
    \tuple{
      l^{\A},  l_{\tuple{0}}^{\A_w}, l_{\tuple{1}}^{\A_w}, \dots l_{\tuple{M}}^{\A_w}
    },  \sigma, \block{t}
    ) =$
    \[
    \begin{cases}
     \delta^{\A}( l^{\A},  \sigma, \block{t}) \times
     \set{l_{\tuple{0}}^{\A_w}, l_{\tuple{1}}^{\A_w}, \dots l_{\tuple{t-1}}^{\A_w}} \times
     \delta^{\A_w}(l_{\tuple{t}}^{\A_w}, \sigma) \times
     \set{l_{\tuple{t+1}}^{\A_w}, \dots l_{\tuple{M}}^{\A_w}} 
     & \text{if } l_{\tuple{0}}^{\A_w} = l_{\tuple{x, k}}^{\A_w},\\
     & \text{ or }
        l_{\tuple{0}}^{\A_w} =  l_k^{\A_w}; \\
        & \text{both with } k +t \leq \tend.\\
        \emptyset & \text{otherwise.}
      \end{cases}
      \]
     This corresponds to the fact that whenever $\A$ has  a delay transition $\sigma, \block{t}$, the
     automaton $\A^{\ddagger}$ must ensure that the $t$-time component of  $\A_w$ takes a transition
     on the symbol $\sigma$, so long as  $\A_w$ has not seen the whole word $w$.
   \item 
       $\delta^{\A^{\ddagger}}(
    \tuple{
      l^{\A},  l_{\tuple{0}}^{\A_w}, l_{\tuple{1}}^{\A_w}, \dots l_{\tuple{M}}^{\A_w}
    }, \tick
    )  $ is as follows.
    \begin{itemize}
      \item
        Let $f(l^{\A})$ defined as
        \[
        f(l^{\A}) = 
        \begin{cases}
          \delta^{\A}( l^{\A}, \tick) & \text{if }  l^{\A} \text{ is not an accepting location of }
          \A\\
          \delta^{\A}( l^{\A}, \tick)  \, \cup\, \set{ l_{\tick}^{\A}} &
          \text{if }  l^{\A} \text{ is an accepting location of }
          \A\\
           l_{\tick}^{\A} &  \text{if }  l^{\A} =  l_{\tick}^{\A}
        \end{cases}
        \]
        This corresponds to the fact that the accepting language of $\A$ does not change if we
        we add a $\tick$ transition from accepting locations to a new sink accepting
        location $l_{\tick}^{\A}$.
        Doing this simplifies the sequel.

      \item 
        Let $g( l_{\tuple{M}}^{\A_w})$ be defined as
        \[
        g( l_{\tuple{M}}^{\A_w}) =
        \begin{cases}
           l_{k+1}^{\A_w} & \text{ if }
            l_{\tuple{M}}^{\A_w} = l_{\tuple{x, k}}^{\A_w}, \text{ or }
            l_{\tuple{M}}^{\A_w} =  l_k^{\A_w}; \text{ both with } k < \tend\\
            l_{\bot}^{\A_w} & \text{otherwise}
          \end{cases}
          \]
          Recall that $ l_{\bot}^{\A_w}$ denotes the location $l_{\tuple{w_{I_\tend}^{\sigma}, w^t_{i_{\tend}}}}$,
          and is the end sink accepting location of $\A_w$.
          Thus, $g( l_{\tuple{M}}^{\A_w})$ denotes the starting location  $l_{k+1}^{\A_w}$ for time $k+1$ if
          $ l_{\tuple{M}}^{\A_w} = l_{\tuple{x, k}}^{\A_w}, \text{ or }
         l_{\tuple{M}}^{\A_w} =  l_k^{\A_w}; \text{ both with } k < \tend$.
         If $ l_{\tuple{M}}^{\A_w}$ corresponds to the $\tend$ locations, then the $g( l_{\tuple{M}}^{\A_w})$
         is the sink accepting location $ l_{\bot}^{\A_w}$.

       \item 
         $\delta^{\A^{\ddagger}}(
    \tuple{
      l^{\A},  l_{\tuple{0}}^{\A_w}, l_{\tuple{1}}^{\A_w}, \dots l_{\tuple{M}}^{\A_w}
    }, \tick
    )  $ is then
    \[
    \begin{cases}
      \emptyset & \left(
        \begin{array}{l}
          \text{if }  l_{\tuple{0}}^{\A_w} \neq l_{\tuple{w_{i_k}^{\sigma}, w_{i_k}^t}}
          \text{ for every } \\
           \qquad \qquad \qquad \qquad 0\leq k\leq \tend \\
          \text{i.e., }  l_{\tuple{0}}^{\A_w} \text{ is not a location in }
          \A_w\\ 
          \text{from which}
          \text{ there exists an}\\
          \text{outgoing }
          \tick \text{ transition}.
        \end{array}
      \right)\\ 
      f(l^{\A}) \times \set{ l_{\tuple{1}}^{\A_w}}\times\set{ l_{\tuple{2}}^{\A_w}}\times
      \dots \set{ l_{\tuple{M}}^{\A_w}}\times \set{g( l_{\tuple{M}}^{\A_w})}
      & 
      \left(
        \begin{array}{l}
          \text{otherwise, i.e. }  l_{\tuple{0}}^{\A_w} \text{ is a location}\\
          \text{in }\A_w \text{ from} 
          \text{ which there exists }\\
          \text{an outgoing }
          \tick \text{ transition}.
        \end{array}
      \right)
    \end{cases}
\]
  % \item 
  %     If $ l_{\tuple{M}}^{\A_w}$ is a location of the form
  %     $l_{\tuple{x, k}}^{\A_w}$ with $k < \tend$,  or
  %     $ l_{\tuple{M}}^{\A_w} = l_k^{\A_w}$  with $k < \tend$, then:
  %     \[
  %     \begin{cases}
  %      \emptyset & \left(
  %      \begin{array}{l}
  %      \text{if }  l_{\tuple{0}}^{\A_w} \neq l_{\tuple{w_{i_k}^{\sigma}, w_{i_k}^t}}
  %      \text{ for every } 0\leq k\leq \tend \\
  %      \text{i.e., }  l_{\tuple{0}}^{\A_w} \text{ is not a location in }
  %      \A_w \text{ from which there is}\\
  %      \text{an outgoing }
  %       \tick \text{ transition}.
  %     \end{array}
  %   \right)\\ 
  %     \delta^{\A}( l^{\A}, \tick) \times
  %      \set{\tuple{l_{\tuple{1}}^{\A_w}, \dots l_{\tuple{M}}^{\A_w}, l_{k+1}^{\A_w}}} & 
  %      \left(
  %      \begin{array}{l}
  %      \text{otherwise,}\\
  %      \text{where }
  %       l_{\tuple{M}}^{\A_w} = l_{\tuple{x, k}}^{\A_w}, \text{ or }
  %        l_{\tuple{M}}^{\A_w} =  l_k^{\A_w}; \text{ both with } k < \tend
  % \end{array}
  %   \right)
  %    \end{cases}
  %    \]
     In this transition, the automaton $\A^{\ddagger}$ checks that for the time segment $\Delta$ corresponding 
     to $\tuple{x,\Delta} = {\tuple{0}}$ or $\Delta =  {\tuple{0}}$, the automaton
     $\A_w$ has finished reading all all the symbols for that time segment.
     If so, the elements in the location tuple of  $\A^{\ddagger}$ are left shifted, and the next
     starting location component corresponding to time $k+1$ added.
     If $k+1 > \tend$, then we just add the looping location $l_{\bot}^{\A_w}$.
     
     Observe that if $l^{\A}$ is an accepting location, and 
     $l_{\tuple{0}}^{\A_w}, l_{\tuple{1}}^{\A_w}, \dots l_{\tuple{M}}^{\A_w}$ all
     correspond to the locations denoting the ends of the corresponding time segments of $w$, and
     $l_{\tuple{M}}^{\A_w} = l_{\bot}^{\A_w}$, then the $\tick$ transitions guarantee that
     the location $\tuple{l^{\A},  l_{\bot}^{\A_w},  l_{\bot}^{\A_w}, \dots,  l_{\bot}^{\A_w}}$
     will be reached.
  
    \end{itemize}

\end{itemize}
%\end{widetext}    
\end{itemize}

It can be checked that the ADB $\A^{\ddagger}$ is such that
\[  
\lan(\A^{\ddagger}) =
\begin{cases}
 w  & \text{if } w\in \lan(\A)\\
 \emptyset & \text{otherwise}
\end{cases}
\]
Thus, $w\in  \lan(\A)$ iff $\lan(\A^{\ddagger}) \neq \emptyset$.
To check for emptiness of $\lan(\A^{\ddagger}) $, we just have to do a reachability search
from the initial location to the accepting locations.
The number of locations of $\A_w $ is  $|w| +\tend$.
The number of locations of $\A^{\ddagger} $ is $n_{\A}\cdot (|w| +\tend)^{M+1}$.
The number of edges of $\A^{\ddagger} $ is $m_{\A} + \tend + |w|$.
To check for reachability of the accepting locations, we do need to explicitly construct the
automaton $\A^{\ddagger}$, if we construct only the reachable from the initial state part using
an on the fly algorithm, we will encounter at most
$\tend + |w| + n_{\A}$ locations  and $m_{\A} + \tend + |w|$ edges.
To traverse an edge, we have to do $O(M)$ work, since
a state of $\A^{\ddagger}$  is a $M+2$-tuple.
Thus, the timed word membership test can be done in time
$O\left(M \cdot \left(n_{\A} + m_{\A} + |w| +  \tend\right)\right)$.
\qed
\end{proof}

\begin{proposition}
\label{proposition:IntersectionEmptiness}
Given three arbitrary ADBs $\A_1$, $\A_2$ and $\A_3$, it is undecidable to check
whether $\lan(\A_1) \cap \lan(\A_2) \cap \lan(\A_3)= \emptyset$.
\end{proposition}
\begin{proof}
 The proof is similar to the proof of the corresponding
  result for context free grammars.
  Let $\T$ be a Turing machine with $\Sigma$ as the tape alphabet,
  $Q$ as the set of locations, and
  $F$ as the set of accepting locations.
  Let a valid computation of $\T$ be denoted by an untimed string
  $w= w_0 \# w_1 \#\dots w_n$ such that $w_0$ represents the initial tape
  configuration, $w_{i+1}$ is a tape configuration that follows from $w_i$ for
  $i\geq 0$, and $\#$ is a special delimiter symbol.
  A valid computation of $\T$ can be denoted by the timed word
  $\widehat{w} = \kappa_0(w_0 \#) \kappa_1(w_1 \#)
  \kappa_2(w_2 \#)  \dots \kappa_n(w_n \#)$.
  That is, $\widehat{w}$ denotes the timed word corresponding to $w$ such that
  the substring  $w_i \#$ occurs at time $i$.
  Observe that a timed word $\widehat{w} = 
  \kappa_0(w_0 \#) \kappa_1(w_1 \#)
  \kappa_2(w_2 \#)  \dots \kappa_n(w_n \#)$
   represents a valid computation of $\T$
  iff all the following conditions hold.
  \begin{enumerate}
  \item
    Each $w_i$ represents a valid configuration of $\T$, that is
    $\kappa_i(w_i \#)$ belongs to the set $(\Sigma\times\set{i})^*
    (Q\times\set{i}) (\Sigma\times\set{i})^* \#$.
  \item
    $w_0$ represents the initial configuration of $\T$, that is
    $\kappa_i(w_0 \#)$ is of the form 
    $(\set{q_0}\times\set{i}) (\Sigma\times\set{i})^* \#$ where
    $q_0$ is the initial location of $\T$.
  \item
    $w_n$ represents an accepting configuration of $\T$, that is
    $\kappa_n(w_n)$ is a string in $(\Sigma\times\set{n})^* 
    (F\times\set{n}) (\Sigma\times\set{i})^* $.
  \item
    Each $w_{i+1}$ is represents a configuration of $\T$ following from
    $w_i$.
  \end{enumerate}

  We can construct an ADB $\A_1$ (having  delay blocks only of delay 0)
  which generates strings satisfying 
  conditions 1,2 and 3.
  We no show that we can construct  ADBs $\A_2$ and $\A_3$ such that
  the strings in $\lan(\A_2) \cap \lan(\A_3)$ satisfy condition 4.
  The automaton $\A_2$ checks that the configuration $w_{2i+1}$ occurring
  at time $2i+1$ follows from the configuration $w_{2i}$ occurring at time
  $2i$ for $i\geq 0$.
  It does this by first generating two symbols from $w_{2i} \#$ at time $2i$,
  and then generating the appropriate first symbol for $w_{2i+1}$ using a delay
  block of delay 1.
  It then repeatedly generates a symbol from  $w_{2i} \#$, and the corresponding
  next symbol for $w_{2i+1}$ using a delay block of delay 1.
  When the $\#$ symbol is generated for $w_{2i} \#$, there can be at most two
  symbols remaining for $w_{2i+1}$ which are then generated.
  This procedure can be repeated for each $i$ by a loop in $\A_2$.
  In addition, there are $\epsilon$-transitions to states which generate
  $w_{2i}$ and stop (this is in case the computation $\widehat{w}$ has on odd
  number of steps.
  A similar automaton $\A_3$ can be constructed which
  ensures that the configuration $w_{2i+2}$ occurring
  at time $2i+2$ follows from the configuration $w_{2i+1}$ occurring at time
  $2i+1$ for $i\geq 0$.
  Thus, a timed word $\widehat{w}$ represents a valid
  configuration of $\T$ iff it belongs to 
  $\lan(\A_1)\cap \lan(\A_2) \cap \lan(\A_3)$ 
  The result follows by noting that  $\T$ accepts a string iff
  $\lan(\A_1)\cap \lan(\A_2) \cap \lan(\A_3) \neq \emptyset$.
  %Given an ADB $\A$, we can check for emptiness by 
  %Proposition~\ref{proposition:Emptiness}.
  %$w_i \#$ occurs at time $i$ for $i\geq 0$.
\qed
\end{proof}

\end{document}